\documentclass[showkeys,floatfix,noshowpacs,preprintnumbers,twocolumn,aps,superscriptaddress,prl,longbibliography]{revtex4-2}

\usepackage[british]{babel}
\usepackage{graphicx}
\usepackage{epsfig}
\usepackage{microtype}
\usepackage{amsfonts}
\usepackage{amsmath}
\usepackage{amssymb}
\usepackage{amsthm}
\usepackage{enumitem}
\usepackage{xcolor}
\usepackage{tikz}
\usetikzlibrary{positioning}
\usetikzlibrary{decorations.pathreplacing}
\usepackage{amsmath,amsthm,amsfonts,amssymb}
\usepackage{thmtools, thm-restate}
\usepackage{braket}
\usepackage{colonequals}
\usepackage{footmisc}
\usepackage{color}
\usepackage{tikz}
\usepackage{scalerel}
\usepackage{comment}
\usepackage{braket}
\usepackage{physics}
\usepackage{bm}

\usepackage{soul}

\usepackage{todonotes}

\theoremstyle{plain}
\newtheorem{theorem}{Theorem}
\newtheorem{proposition}[theorem]{Proposition}

\newtheorem{corollary}[theorem]{Corollary}

\theoremstyle{definition}

\newtheorem{question}{Question}
\theoremstyle{remark}

\newtheorem*{remark*}{Remark}

\usepackage{braket}
\usepackage{colonequals}
\newcommand{\defeq}{\colonequals}

\usepackage[colorlinks=true,linkcolor=blue,citecolor=magenta,urlcolor=blue]{hyperref}
\usepackage{cleveref}
\newcommand{\stress}{\textit}
\newcommand{\Vorobev}{Vorob{\textquotesingle}ev}
\newcommand{\ie}{i.e.~}
\newcommand{\eg}{e.g.~}

\newcommand{\clS}{\bar{S}}

\newcommand{\psection}[1]{\paragraph*{#1.---}\hspace{-1em}}

\begin{document}

\title{Contextuality with Pauli observables in cycle scenarios}


\newcommand{\inllong}{INL -- International Iberian Nanotechnology Laboratory, Av.~Mestre Jos\'e Veiga s/n, 4715-330 Braga, Portugal}
\newcommand{\inlshort}{INL -- International Iberian Nanotechnology Laboratory, Braga, Portugal}
\newcommand{\haslablong}{HASLab, INESC TEC, Universidade do Minho, Departamento de Informática, Campus de Gualtar, 4710-057 Braga, Portugal}
\newcommand{\haslabshort}{HASLab, INESC TEC, Universidade do Minho, Braga, Portugal}
\newcommand{\diumlong}{Departamento de Inform\'atica, Universidade do Minho, Campus de Gualtar, 4710-057 Braga, Portugal}
\newcommand{\diumshort}{Departamento de Inform\'atica, Universidade do Minho, Braga, Portugal}

\author{Raman Choudhary}
\affiliation{\inlshort}
\affiliation{\haslabshort}
\affiliation{\diumshort}

\author{Rui Soares Barbosa}
\affiliation{\inlshort}

\begin{abstract}
Contextuality is a fundamental marker of quantum non-classicality, which
has emerged as a key resource for quantum computational advantage in multiple settings.
Many such results hinge specifically on contextuality witnessed through Pauli measurements.
In this work, we initiate a systematic study of (state-dependent) Pauli contextuality,
focusing on cycle measurement scenarios, the simplest scenarios capable of exhibiting contextual behaviour.
First, we study realizability of cycle scenarios with multi-qubit Pauli observables:
we show that the maximum size of a cycle faithfully realizable by $m$-qubit Paulis is upper bounded by $3m$,
while we construct explicit realizations of cycles of size $2m-1$ or $2m$, depending on whether $m \not\equiv 1 \pmod{3}$ or $m \equiv 1 \pmod{3}$.
Then, we investigate the presence of contextuality:
we prove that no $n$-cycle Pauli realization for $n > 4$ can witness contextuality (on any quantum state),
whereas for $n = 4$ every Pauli realization exhibits contextuality, attaining the quantum bound for all noncontextuality inequalities on some pure state.
Finally, we discuss arbitrary Pauli scenarios in light of \Vorobev’s theorem, and show that,
contrary to what our cycle characterization might suggest, the presence of $4$-cycles is not necessary for witnessing contextuality in general Pauli scenarios.
\end{abstract}
\maketitle

\psection{Introduction}%
Quantum computational advantage must rely on features of quantum theory that cannot be simulated classically.
Contextuality -- the impossibility of regarding quantum measurements as revealing predetermined outcomes --
stands out as a fundamental, rigorous marker of non-classicality,
with mounting evidence for its role as a resource for quantum speed-up in a variety of settings,
ranging from universal quantum computation via magic state distillation \cite{howard2014contextuality} to measurement-based quantum computing \cite{raussendorf2013contextuality} to variational quantum algorithms \cite{kirby2019contextuality}, among others.
Notably, many of these results boil down to the presence of contextuality witnessed through (qudit or qubit) Pauli observables.
After all, Paulis form the backbone of many models of quantum computation.

Pauli observables have also been studied in quantum foundations and, specifically in the qubit case,
used to provide especially simple and appealing proofs of contextuality,
namely proofs of state-independent \stress{strong} contextuality, with a powerful logical flavour and applicable to \stress{any} quantum state.
Two of the simplest are given by the Peres--Mermin square \cite{peres1990incompatible} and the Mermin star \cite{mermin1990simple}.
But while strong contextuality relative to $m$-qubit Pauli observables has been extensively studied
-- from such state-independent proofs to state-dependent arguments such as GHZ \cite{greenberger1989going,greenberger1990bell}, or more generally, All-versus-Nothing arguments \cite{abramsky2015ccp,abramsky2017avn} --
the \stress{weaker} notion of probabilistic contextuality with Pauli measurements has received comparatively little attention,
even though it does play a crucial role in many of the applications.

Here, we initiate a systematic study of state-dependent Pauli contextuality.
For reasons that we further expound below,
we focus on the class of scenarios whose measurement compatibility structure is given by a cycle graph.
These are the simplest and most fundamental scenarios capable of displaying contextual behaviour.
Our main result provides a complete characterization of Pauli contextuality in such cycle scenarios,
revealing a striking dichotomy:
while Pauli realizations of the 4-cycle (a.k.a.~CHSH) scenario always achieve the Tsirelson bound, the maximal possible quantum violation of the associated noncontextuality inequalities,
no larger cycle scenario can witness contextuality using Pauli measurements.
This finding is complemented by several observations regarding the faithful realizability of cycle scenarios using Pauli operators on a fixed number $m$ of qubits, including an impossibility result establishing an upper bound of $3m$ for the cycle size and an explicit construction of cycles of size $2m$ or $2m-1$.

Together, these results illuminate fundamental structural constraints on Pauli-based contextuality.
We conclude with a brief examination of more general Pauli scenarios, demonstrating through a counterexample
that the presence of 4-cycles is not a necessary condition for contextuality in that broader setting,
so that the full landscape of Pauli contextuality is yet to be mapped.

\psection{Further motivation}
Let us delve a little deeper into the motivation.
Pauli measurements play a central role in many models of quantum computation, in both near-term and fault-tolerant settings.
On the one hand, they are of course essential in quantum error correction via stabiliser codes \cite{gottesman1997stabilizer}.
Also, Pauli-based computation, a recently proposed quantum universal model,
involves sequential measurements of adaptively chosen commuting Pauli observables on an initially prepared magic state \cite{bravyi2016trading}.
On the other hand, hybrid algorithms like variational quantum eigensolvers (VQEs), seen as strong candidates for witnessing practical quantum advantage in NISQ devices, also involve measurements of Pauli observables in each iteration.
Focusing on the measurement part of VQE, Kirby and Love~\cite{kirby2019contextuality} defined a (state-independent) contextuality test on VQE instances, and
subsequently~\cite{kirby2020classical} showed that the classical simulation of noncontextual VQEs is \stress{only} an NP-complete problem,
while simulating arbitrary VQE instances is QMA-hard.
That is, noncontextual instances are only classically hard, as opposed to \stress{quantumly} hard.

The contextuality test for VQE instances proposed in Ref.~\cite{kirby2019contextuality} concerns
state-independent contextuality.
The instances considered in that work are given by Hamiltonians expressed in Pauli decomposition, \ie as a linear combination of operators from the $m$-qubit Pauli group $\mathcal{P}_m$.
Starting from the set $S \subseteq \mathcal{P}_m$ of Paulis
appearing in such a decomposition, it considers its closure $\clS$ under products of commuting elements \footnote{$\clS$ is the partial Abelian group (in the same sense of partial Boolean algebra) generated by $S$ \stress{within} $\mathcal{P}_m$; we refer to upcoming work for more about this perspective.}.
The instance is deemed to be contextual if $\clS$ 
exhibits state-independent strong contextuality (like the above-mentioned square and star).
The authors also provide a necessary and sufficient condition on the set $S$ for this to happen.
The initial motivation for the present work comes from noticing the following caveat in this analysis.
The VQE algorithm performs only the measurements in $S$, not those in $\clS$.
As such, in a classical simulation, it suffices to reproduce the statistics of these measurements on certain quantum states,
regardless of whether they could be consistently extended to encompass $\clS$.
A more refined classicality analysis
would thus focus on the (weak) contextuality of the correlations
realizable by the measurements of $S$ alone on some quantum state.

Prompted by the above discussion, we consider the following question:
\textit{given a set $S$ of $m$-qubit Pauli operators, when does it \stress{directly} witness contextuality?}

One way to approach this question is to take the compatibility graph of
the Pauli operators in $S$ and treat it as an abstract measurement scenario.
One could then derive noncontextuality inequalities for this scenario and test for violations with the given Pauli realization.
A major bottleneck, however, is that finding all the noncontextuality inequalities for an arbitrary scenario encoded as a compatibility graph is an NP-complete problem.

Hence, as a way to get a handle on the general problem, we start by considering a particular class of measurement scenarios:
those where the measurement compatibility structure is given by a cycle graph $C_n$. Two main reasons inform this choice:
\begin{enumerate}[leftmargin=*]
 \setlength{\topsep}{0pt}
  \setlength{\parskip}{0pt}
    \item All the noncontextuality inequalities for these scenarios have been explicitly characterised by Ara\'ujo et al.~\cite{araujo2013all}. In fact, these are the only fully characterised non-Bell-type scenarios to date.
    \item \Vorobev's theorem \cite{vorobev1962consistent}, phrased in terms of compatibility graphs, requires the presence of at least one cycle of size $n \geq 4$ as an \textit{induced} subgraph for a scenario to admit any contextual correlations. Cycle graphs are thus the minimal measurement scenarios that may exhibit contextuality.
\end{enumerate}

\psection{Contributions}
Focusing on such cycle scenarios, we pose the two questions that we address in this work.
First, which $n$-cycle scenarios are faithfully realizable by $m$-qubit Pauli operators (for each fixed $m$)
\footnote{A faithful realization of a graph $G=(V_G,E_G)$ by $m$-qubit Paulis is an assignment $P \colon V_G \to \mathcal{P}_m$ of $m$-qubit Pauli operators to vertices of the graph such that $P(v)$ and $P(w)$ commute \textit{if and only if} $v=w$ or $\{v,w\}\in E$.
When $G$ is a cycle of size at least $4$ a realization must necessarily be injective. Hence, it amounts to an isomorphic copy of $G$ as an induced subgraph of the compatibility graph of $\mathcal{P}_m$.}?
Second, which such Pauli cycles witness contextuality on some quantum state?

Addressing the first question,
we prove a no-go (impossibility) theorem that provides a cycle size, as a function of $m$, above which no cycle can be faithfully realized by $m$-qubit Pauli operators. We also ran sub-graph isomorphism tests to find out that the upper bound in the no-go theorem is not always tight.
This directed our focus to the realizablility of cycles below that bound. 
We provide a construction that guarantees faithful Pauli realization of all cycles up to a size given as a function of $m$.
This also produces realizations of the largest cycles we could achieve for generic $m$, considerably improving (approximately by a factor of $2$) on another general construction found in the literature.

Moving on to the second question, we fully characterized $n$-cycle contextuality with $m$-qubit Paulis, showing that only 4-cycles can witness contextuality, and that those do so maximally, reaching Tsirelson's bound. 
Thus, owing to \Vorobev’s theorem, which guarantees that any contextuality-witnessing
scenario contains an induced $n$-cycle sub-scenario, we provide a sufficient but not necessary condition to witness contextuality in an arbitrary KS scenario realizable by $m$-qubit Pauli operators.

In the remainder of the text, we provide a concise summary of our main results and proofs, which give the flavour of the key methods used. Full details can be found in the corresponding appendices.

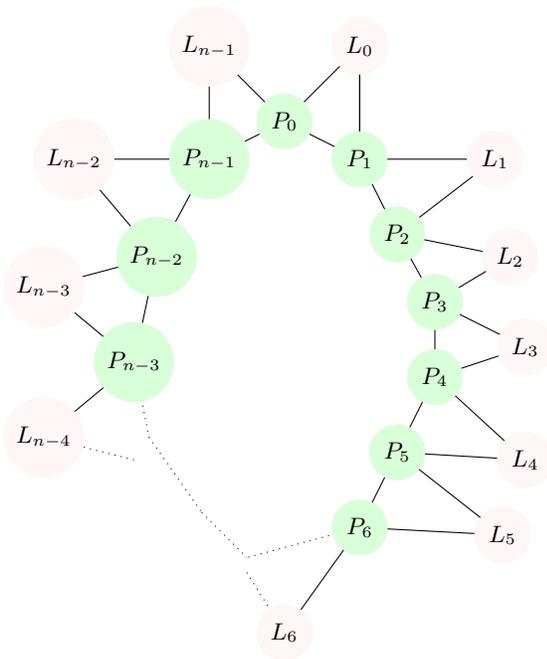
\begin{figure}
\begin{center}
\begin{tikzpicture}
[roundnode1/.style={circle, fill=green!15, minimum size=3mm}, roundnode2/.style={circle, fill=pink!15, minimum size=3mm}]
\draw (0,0)node[roundnode1]{$P_0$}--(1.0,-0.5)node[roundnode1]{$P_1$};
\draw (1.0,-0.5)node[roundnode1]{$P_1$}--(1.5,-1.5)node[roundnode1]{$P_2$};
\draw (1.5,-1.5)node[roundnode1]{$P_2$}--(2.0,-2.4)node[roundnode1]{$P_3$};
\draw (2.0,-2.4)node[roundnode1]{$P_3$}--(2.0,-3.4)node[roundnode1]{$P_4$};
\draw (2.0,-3.4)node[roundnode1]{$P_4$}--(1.5,-4.4)node[roundnode1]{$P_5$};
\draw (1.5,-4.4)node[roundnode1]{$P_5$}--(1.0,-5.4)node[roundnode1]{$P_6$};
\draw [dotted](1.0,-5.4)--(-0.5,-5.8);
\draw (0,0)node[roundnode1]{$P_0$}--(-1.0,-0.5)node[roundnode1]{$P_{k-1}$};
\draw (-1.0,-0.5)node[roundnode1]{$P_{n-1}$}--(-1.7,-1.8)node[roundnode1]{$P_{n-2}$};
\draw (-1.7,-1.8)node[roundnode1]{$P_{n-2}$}--(-2.0,-3.2)node[roundnode1]{$P_{n-3}$};
\draw [dotted](-2.0,-3.2)--(-1.8,-4.2);
\draw [dotted](-1.8,-4.2)--(-1.1,-5.2);
\draw [dotted](-1.1,-5.2)--(-0.5,-5.8);
\draw (0,0)node[roundnode1]{$P_0$}--(-1.0,1.0)node[roundnode2]{$L_{n-1}$};
\draw (-1.0,-0.5)node[roundnode1]{$P_{n-1}$}--(-1.0,1.0)node[roundnode2]{$L_{n-1}$};
\draw (-1.0,-0.5)node[roundnode1]{$P_{n-1}$}--(-2.8,-0.5)node[roundnode2]{$L_{n-2}$};
\draw (-1.7,-1.8)node[roundnode1]{$P_{n-2}$}--(-2.8,-0.5)node[roundnode2]{$L_{n-2}$};
\draw (-1.7,-1.8)node[roundnode1]{$P_{n-2}$}--(-3.2,-2.2)node[roundnode2]{$L_{n-3}$};
\draw (-2.0,-3.2)node[roundnode1]{$P_{n-3}$}--(-3.2,-2.2)node[roundnode2]{$L_{n-3}$};
\draw (-2.0,-3.2)node[roundnode1]{$P_{n-3}$}--(-3.2,-4.2)node[roundnode2]{$L_{n-4}$};
\draw [dotted](-2.0,-4.5)--(-3.2,-4.2)node[roundnode2]{$L_{n-4}$};
\draw (0,0)node[roundnode1]{$P_0$}--(1,1)node[roundnode2]{$L_0$};
\draw (1.0,-0.5)node[roundnode1]{$P_1$}--(1,1)node[roundnode2]{$L_0$};
\draw (1.0,-0.5)node[roundnode1]{$P_1$}--(1,1)node[roundnode2]{$L_0$};
\draw (1.0,-0.5)node[roundnode1]{$P_1$}--(2.8,-0.5)node[roundnode2]{$L_1$};
\draw (1.5,-1.5)node[roundnode1]{$P_2$}--(2.8,-0.5)node[roundnode2]{$L_1$};
\draw (1.5,-1.5)node[roundnode1]{$P_2$}--(3.0,-1.8)node[roundnode2]{$L_2$};
\draw (2.0,-2.4)node[roundnode1]{$P_3$}--(3.0,-1.8)node[roundnode2]{$L_2$};
\draw (2.0,-2.4)node[roundnode1]{$P_3$}--(3.2,-3.0)node[roundnode2]{$L_3$};
\draw (2.0,-3.4)node[roundnode1]{$P_4$}--(3.2,-3.0)node[roundnode2]{$L_3$};
\draw (2.0,-3.4)node[roundnode1]{$P_4$}--(3.2,-4.5)node[roundnode2]{$L_4$};
\draw (1.5,-4.4)node[roundnode1]{$P_5$}--(3.2,-4.5)node[roundnode2]{$L_4$};
\draw (1.5,-4.4)node[roundnode1]{$P_5$}--(2.9,-5.5)node[roundnode2]{$L_5$};
\draw (1.0,-5.4)node[roundnode1]{$P_6$}--(2.9,-5.5)node[roundnode2]{$L_5$};
\draw (1.0,-5.4)node[roundnode1]{$P_6$}--(0.0,-6.8)node[roundnode2]{$L_6$};
\draw [dotted](-0.5,-6)--(0.0,-6.8)node[roundnode2]{$L_6$};
\end{tikzpicture}
\caption[5]{Construction used in proving the no-go theorem and in characterizing Pauli contextuality for $n$-cycle scenarios. Arbitrary Pauli $n$-cycle (green) and the constructed edge Paulis $L_i=P_{i}P_{i\oplus 1}$ (pink). Each edge Pauli commutes with all other edge Paulis except the nearest and next-nearest neighbours on both sides. 
}
\label{fig1}
\end{center}
\end{figure}

\begin{figure}
\begin{center}
\begin{tikzpicture}
[roundnode1/.style={circle, fill=green!15, minimum size=3mm}, roundnode2/.style={circle, fill=pink!15, minimum size=3mm}]
\draw (0,0)node[roundnode1]{$XI$}--(0,2)node[roundnode1]{$IX$};
\draw (0,2)node[roundnode1]{$IX$}--(-2,1)node[roundnode1]{$YX$};
\draw (-2,1)node[roundnode1]{$YX$}--(-2,-1)node[roundnode1]{$ZY$};
\draw (-2,-1)node[roundnode1]{$ZY$}--(0,-2)node[roundnode1]{$IY$};
\draw (0,-2)node[roundnode1]{$IY$}--(0,0)node[roundnode1]{$XI$};
\draw (0,2)node[roundnode1]{$IX$}--(2,1)node[roundnode1]{$ZX$};
\draw (2,1)node[roundnode1]{$ZX$}--(2,-1)node[roundnode1]{$YY$};
\draw (2,-1)node[roundnode1]{$YY$}--(0,-2)node[roundnode1]{$IY$};

\end{tikzpicture}
\caption[5]{Counterexample falsifying the conjecture that a 4-cycle is necessary for an arbitrary scenario to witness Pauli contextuality: Pauli realization of two conjoined 5-cycles, which violates a non-cycle noncontextuality inequality.}
\label{fig2}
\end{center}
\end{figure}
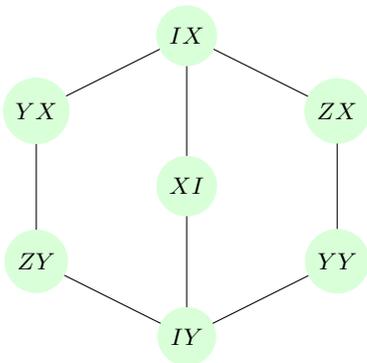

\psection{Impossible Pauli cycles}
We first consider the question of establishing which $n$-cycle scenarios are faithfully realizable by $m$-qubit Pauli operators.
To build some intuition on this, we ran subgraph isomorphism tests within the compatibility graphs of $\mathcal{P}_m$ for small qubit number $m$.
We found that $2$-qubit Pauli operators cannot realize cycles beyond size 6, while $3$-qubit Paulis cannot realize cycles beyond size 9.
This suggested a general investigation to identify an upper bound on the size of cycles admitting $m$-qubit Pauli realizations.

We derived a no-go theorem showing it is impossible to faithfully realize an $n$-cycle scenario with $m$-qubit Pauli operators when $n > 3m$. This bound is clearly tight for $m = 2,3$, but continuing our computational search further revealed that it is not tight for $m = 4$, where the largest faithfully realizable cycle has size 9 and not 12.
Similarly, for $m = 5$, the largest faithfully realizable cycle is the 12-cycle, instead of the 15-cycle scenario allowed in principle by the no-go theorem.
We also observe that, for a fixed $m$, some values of $n$ are `skipped': \eg there is a $9$-cycle of $3$-qubit Paulis, but no $8$-cycle.

The proof of our no-go theorem goes as follows.
We start off with an arbitrary realization $\{P_i\}_{i=0}^{n-1}$ of an $n$-cycle with $m$-qubit Paulis and extend it through the construction depicted in \Cref{fig1}.
For each edge in the cycle, we consider an additional ``edge'' Pauli, $L_i = P_i P_{i \oplus 1}$ \footnote{These new edge Paulis are not taken to be part of the contextuality scenario; they are simply a device to help derive a contradiction.}. 
Crucially, whether each $L_i$ commutes or anti-commutes with the original cycle Paulis and with the other edge Paulis
can be determined by a general rule, which does not depend on the concrete Pauli realization we started from.
Namely, a Pauli $Q$ commutes with $L_i$ if and only if it commutes with both $P_i$ and $P_{i\oplus 1}$ or with neither.
In other words, the neighbourhood of $L_i$ in the compatibility graph of Pauli operators is the exclusive disjunction
between those of $P_i$ and of $P_{i\oplus 1}$.
More specifically for cycle realizations, each edge Pauli commutes with all other edge Paulis except for its nearest and next-nearest neighbours on either side.

We use this commutativity structure
to show that the set $S = \{P_0,L_0,L_3,L_6,\ldots,L_{3\lceil\frac{n}{3}\rceil-6}\} = \{P_0\} \cup \{L_{3j}\}_{j=0}^{\lceil\frac{n}{3}\rceil-2}$
consists of $\lceil\frac{n}{3}\rceil$ independent, pairwise-commuting Paulis.
Since the maximum such set within the $m$-qubit Pauli group $\mathcal{P}_m$ has size $m$, we conclude that
$\lceil\frac{n}{3}\rceil \leq m$, hence $n \leq 3m$.
This shows that any cycle faithfully realizable by $m$-qubit Paulis has size at most $3m$.

Full details of this proof can be found in \Cref{Impossible realizations}. 


\psection{Constructing Pauli cycle realizations}
The no-go theorem precludes $m$-qubit Pauli realizations of cycles beyond size $3m$.
Complementarily, we now focus on constructing some possible Pauli realizations.
A general construction by Abramsky, Cercelescu, and Constantin~\cite{abramsky2024commutation}
yields a realization of an $m$-cycle with $m$-qubit Paulis.
We considerably improve on this by constructing a faithful realization of a cycle of size $2m-1$ or $2m$, depending on whether $m \not\equiv 1$ (mod 3) or $m \equiv 1$ (mod 3).
The same construction further guarantees the existence of faithful realizations of all cycles with size ranging from 3 to $2m-3$ or to $2m-2$, again depending on whether $m \not\equiv 1$ (mod 3) or $m \equiv 1$ (mod 3).

We base this construction on two more general results:
\begin{enumerate}[label=(\roman*),leftmargin=*]
 \setlength{\topsep}{0pt}
  \setlength{\parskip}{0pt}
\item faithful realizations of the path graphs $H_l$ and $H_{l'}$, with $m$- and $m'$-qubit Paulis respectively,
can be concatenated 
to produce a faithful realization of the larger path graph $H_{l+l'- 2}$ with $(m+m')$-qubit Paulis;
\item any faithful  the path graph $H_l$ with $m$-qubit Paulis can be extended to a faithful realization of the cycle $C_l$ with $(m+1)$-qubit Paulis.
\end{enumerate}
These constructions are depicted in \Cref{Path_graph_combination,Path_graph_to_cycle}, respectively.


Repeated application of (i) starting with a concrete realization for $m = 3$ and $l = 8$ yields
a faithful realization of the path graph $H_{6 p + 2}$ using $3p$-qubit Paulis for any $p \geq 1$.
This is then converted via (ii) to a faithful realization of the $(6 p + 2)$-cycle using $3p+1$-qubit Paulis.
This iterated construction constitutes the basis of our results: by a straightforward change of variable to fix the number of qubits $m$, one obtains the faithful realizations claimed above.

More details on our constructions of faithful cycle realizations can be found in \Cref{Actualizing faithful realizations}.

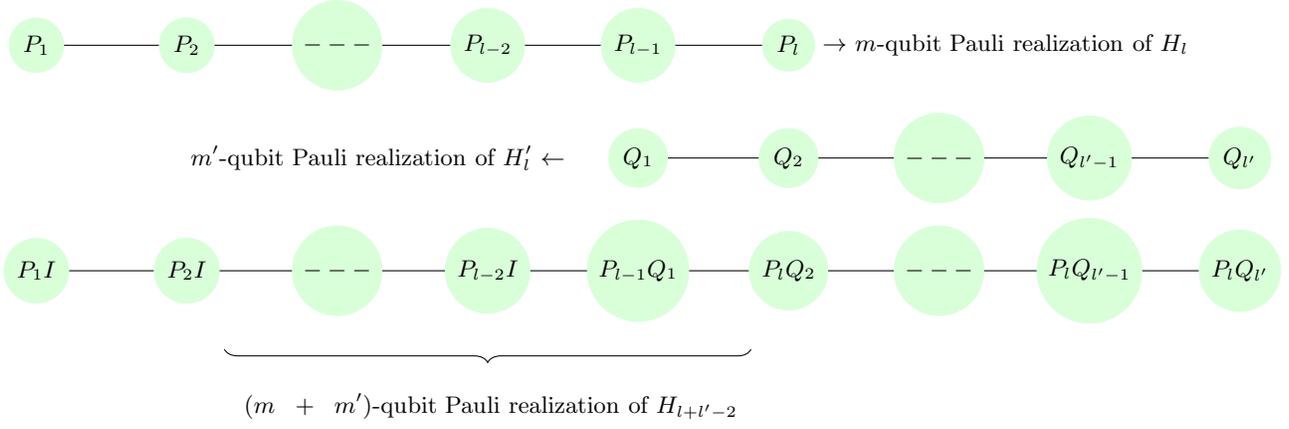
\begin{figure*}
\begin{center}
\begin{tikzpicture}
[roundnode1/.style={circle, fill=green!15, minimum size=05mm}, roundnode2/.style={circle, fill=pink!15, minimum size=5mm}]
\node[text width=5.5cm] at (3.2,0) 
    {$\rightarrow$ $m$-qubit Pauli realization of $H_l$};
\node[text width=5.5cm] at (-5.2,-1.5) 
    {$m'$-qubit Pauli realization of $H_l'$ $\leftarrow$};
\draw [decorate,decoration={brace,amplitude=5pt,raise=4ex}]
  (-0.5,-3.5) -- (-7.5,-3.5);
\node[text width=10cm] at (-2.2,-4.8) 
    {$(m+m')$-qubit Pauli realization of $H_{l+l'-2}$}; 
\begin{scope}
\draw (-10,0)node[roundnode1]{$P_1$}--(-8,0)node[roundnode1]{$P_2$};
\draw (-8,0)node[roundnode1]{$P_2$}--(-6,0)node[roundnode1]{$---$};
\draw (-6,0)node[roundnode1]{$---$}--(-4,0)node[roundnode1]{$P_{l-2}$};
\draw (-4,0)node[roundnode1]{$P_{l-2}$}--(-2,0)node[roundnode1]{$P_{l-1}$};
\draw (-2,0)node[roundnode1]{$P_{l-1}$}--(0,0)node[roundnode1]{$P_{l}$};
\end{scope}
\begin{scope}
\draw (-2,-1.5)node[roundnode1]{$Q_1$}--
(0,-1.5)node[roundnode1]{$Q_2$};
\draw (0,-1.5)node[roundnode1]{$Q_2$}--
(2,-1.5)node[roundnode1]{$---$};
\draw (2,-1.5)node[roundnode1]{$---$}--
(4,-1.5)node[roundnode1]{$Q_{l'-1}$};
\draw (4,-1.5)node[roundnode1]{$Q_{l'-1}$}--
(6,-1.5)node[roundnode1]{$Q_{l'}$};
\end{scope}
\begin{scope}
\draw (-10,-3)node[roundnode1]{$P_1I$}--(-8,-3)node[roundnode1]{$P_2I$};
\draw (-8,-3)node[roundnode1]{$P_2I$}--(-6,-3)node[roundnode1]{$---$};
\draw (-6,-3)node[roundnode1]{$---$}--(-4,-3)node[roundnode1]{$P_{l-2}I$};
\draw (-4,-3)node[roundnode1]{$P_{l-2}I$}--(-2,-3)node[roundnode1]{$P_{l-1}Q_1$};
\draw (-2,-3)node[roundnode1]{$P_{l-1}Q_1$}--
(0,-3)node[roundnode1]{$P_lQ_2$};
\draw (0,-3)node[roundnode1]{$P_lQ_2$}--
(2,-3)node[roundnode1]{$---$};
\draw (2,-3)node[roundnode1]{$---$}--
(4,-3)node[roundnode1]{$P_lQ_{l-1}$};
\draw (4,-3)node[roundnode1]{$P_lQ_{l'-1}$}--
(6,-3)node[roundnode1]{$P_lQ_{l'}$};
\end{scope}
\end{tikzpicture}
\caption{Concatenating a faithful $m$-qubit Pauli realization of the path graph $H_l$ with a faithful $m'$-qubit Pauli realization of $H_{l'}$ to obtain a ($m+m'$)-qubit Pauli realization of $H_{l+l'-2}$.
This works by appending $m'$-qubit identity operators in the first $l-2$ operators of the realization of path graph $H_l$ while introducing the Pauli realization for $H_{l'}$ from the $(l-1)$th node in $H_l$. In the additional $l'-2$ nodes, one then fixes $P_l$ in the first $m$ qubits.} 
\label{Path_graph_combination}
\end{center}
\end{figure*} 

\begin{figure*}
\begin{center}
\begin{tikzpicture}
[roundnode1/.style={circle, fill=green!15, minimum size=05mm}, roundnode2/.style={circle, fill=pink!0, minimum size=0mm}]
\node[text width=5.5cm] at (3.2,0) 
    {$\rightarrow$ $m$-qubit Pauli realization of $H_l$};
\node[text width=6.8cm] at (4,-1.5) 
    {$\rightarrow$ $(m+1)$-qubit Pauli realization of         $   C_l$};    
\draw (-10,0)node[roundnode1]{$P_1$}--(-8,0)node[roundnode1]{$P_2$};
\draw (-8,0)node[roundnode1]{$P_2$}--(-6,0)node[roundnode1]{$---$};
\draw (-6,0)node[roundnode1]{$---$}--(-4,0)node[roundnode1]{$P_{l-2}$};
\draw (-4,0)node[roundnode1]{$P_{l-2}$}--(-2,0)node[roundnode1]{$P_{l-1}$};
\draw (-2,0)node[roundnode1]{$P_{l-1}$}--(0,0)node[roundnode1]{$P_{l}$};
\begin{scope}
\draw (-10,-1.5)node[roundnode1]{$XP_1$}--(-8,-1.5)node[roundnode1]{$IP_2$};
\draw (-8,-1.5)node[roundnode1]{$IP_2$}--(-6,-1.5)node[roundnode1]{$---$};
\draw (-6,-1.5)node[roundnode1]{$---$}--(-4,-1.5)node[roundnode1]{$IP_{l-2}$};
\draw (-4,-1.5)node[roundnode1]{$IP_{l-2}$}--(-2,-1.5)node[roundnode1]{$IP_{l-1}$};
\draw (-2,-1.5)node[roundnode1]{I$P_{l-1}$}--(0,-1.5)node[roundnode1]{$YP_{l}$};
\draw (0,-1.5)node[roundnode1]{$YP_{l}$}--(0,-2.5);
\draw (0,-2.5)--(-10,-2.5);
\draw (0,-2.5)--(-10,-2.5);
\draw (-10,-2.5)--(-10,-1.5)node[roundnode1]{$XP_1$};
\end{scope}
\end{tikzpicture}
\caption{Converting a faithful $m$-qubit Pauli realization of path graph $H_l$ to a faithful $(m+)1$-qubit Pauli realization of $C_l$. This works by appending a single-qubit identity operator everywhere except the first and the last vertices of $H_l$, whose operators anticommute with each other. To make them commute, we choose two anti-commuting single-qubit Pauli operators, such as $X$ and $Y$, and append one to each of the two vertices, as shown.} 
\label{Path_graph_to_cycle}
\end{center}
\end{figure*}
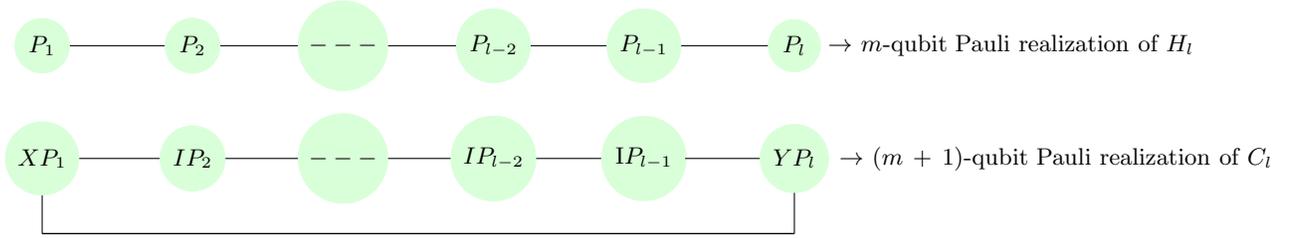

\psection{Pauli contextuality in cycle scenarios}
We now turn our attention to the question of witnessing contextuality in $n$-cycle scenarios with $m$-qubit Pauli measurements.
We provide a complete characterization, which splits into two extreme cases.
For $n>4$, no faithful $n$-cycle of multi-qubit Paulis can witness contextuality: every quantum state yields noncontextual correlations.
By contrast, for $n=4$, every faithful Pauli realization not only exhibits contextuality but even reaches the quantum bound (also known as Tsirelson's bound) for all $4$-cycle non-contextuality inequalities: for each such inequality, there exists a pure quantum state attaining the quantum bound.

Our analysis is enabled by the complete characterization of the non-contextual polytopes of $n$-cycle scenarios in Ref.~\cite{araujo2013all}. 
All the facet-defining non-contextuality inequalities for the $n$-cycle scenario have the form
\begin{equation}\label{eq:NCineq}
\textstyle\sum\nolimits_{i=0}^{n-1}\gamma_{i}\langle A_{i}A_{i\oplus 1}\rangle \overset{\textit{NCHV}}{\leq} n- 2
\end{equation}
where $\gamma_{i} \in \{-1,+1\}$ with $\prod_{i=0}^{n-1} \gamma_i = -1$ and $\{A_i\}_{i=0}^{n-1}$ are the measurements corresponding to each vertex in the cycle.

For a Pauli realization, each term of the sum on the left corresponds to an edge Pauli $L_i$, again as in \Cref{fig1}, so that the quantity being bounded is the expectation value of the operator
\[\Gamma \defeq \sum_{i=0}^{n-1}\gamma_{i} L_i \text{.}\]
Squaring $\Gamma$ and using simple algebra and the commutativity properties of edge Paulis,
we conclude:
\begin{itemize}[leftmargin=*]
 \setlength{\topsep}{0pt}
  \setlength{\parskip}{0pt}
    \item for $n>4$, that $\langle \Gamma^2\rangle \leq n^2 - 4n < (n-2)^2$,
    from which it follows that the inequality in \eqref{eq:NCineq} is not violated;
    \item for $n=4$, that $\Gamma^2 = 4(I+\gamma_1\gamma_3 P_0P_1P_2P_3)$, where $\gamma_1\gamma_3 \in \{-1,+1\}$, so that any $\gamma_1\gamma_3$-eigenstate $\ket{\psi}$ of $P_0P_1P_2P_3$ (itself a Pauli operator) yields $\langle \Gamma^2 \rangle_{\ket{\psi}} = 8$ and $\langle \Gamma \rangle_{\ket{\psi}} = \pm 2\sqrt{2}$.
\end{itemize}

Details of these derivations can be found in \Cref{sec:contextuality_pauli}:
in particular, the $n=4$ case is analysed in \Cref{n=4}, the $n = 5$ case in \Cref{n=5}, and this is then generalised to arbitrary $n > 4$ in \Cref{n>4}.

\psection{Pauli contextuality in scenarios beyond cycles}
\Vorobev's theorem \cite{vorobev1962consistent} implies that any contextuality-witnessing scenario contains an $n$-cycle with $n \geq 4$ as an induced sub-scenario.
In light of our complete characterization of Pauli contextuality for $n$-cycle scenarios, a natural question ensues:
is the presence of a $4$-cycle as an induced sub-scenario necessary and sufficient for an arbitrary Kochen--Specker scenario to witness contextuality with multi-qubit Pauli measurements?

We answer this question negatively.
While the condition is clearly sufficient given our previous results,
we show that it is not necessary by exhibiting a counterexample.
We provide a scenario with no induced $4$-cycles but admitting a $2$-qubit Pauli realization that violates a noncontextuality inequality.
\Cref{fig2} depicts its compatibility graph, together with the Pauli realization.
The scenario consists of two $5$-cycles glued along two edges.

We found new non-cycle inequalities for this scenario using the \textsc{PORTA} package \cite{porta}.
One such inequality is written, in operator-theoretic form and for this particular Pauli realization, as
\begin{equation}
    \begin{split}
     -(ZI+XZ+YI+XY+YI)+ ZI - \\(IY + XI + ZY + YX) \leq 4I .   
    \end{split}
\end{equation}
The maximum eigenvalue of the operator sum on the left-hand side turns out to be  $4.2716 > 4$,
implying that the inequality is violated for some (eigen)state.

This violation happens despite the fact that individual cycles in the scenario cannot exhibit contextuality by themselves,
offering a counterexample to the conjecture implicit in the question above.
One may conclude that, while the presence of cycles in arbitrary scenarios is required to impede \Vorobev's inductive construction
of a probability distribution on joint outcomes for all measurements,
it is not the case that every instance of contextuality can be witnessed within a single cycle.

Further details are provided in \Cref{Arbitrary_cases}.

\bibliographystyle{plain}
\let\oldaddcontentsline\addcontentsline
\renewcommand{\addcontentsline}[3]{}
\bibliography{references}
\let\addcontentsline\oldaddcontentsline

\clearpage
\appendix
\setcounter{secnumdepth}{2}

\tableofcontents

\section{Constraints on faithful Pauli realizations of cycles}\label{Constraints of faithful Pauli realizations of cycles}
Before discussing our results on qubit Pauli contextuality and faithful realization of cycle scenarios by Pauli operators, we start by describing the structural constraints 
enforced on any faithful Pauli realization of a cycle scenario. These constraints are a consequence of the properties of qubit Pauli operators in the Pauli group $\mathbf{\mathcal{P}_m}$ as well as the faithfulness condition on the realization. These constraints drive all the analyses to be followed later on in the paper.  

We start with some basic notational clarification. Throughout the rest of the work,
an $m$-qubit Pauli operators means an operator composed of tensor product of $m$ single qubit operators: the identity operator and/or the following Pauli operators:
$$
X = 
\begin{bmatrix}
     0 & 1 \\ 1 & 0
\end{bmatrix},\;
Y = 
\begin{bmatrix}
     0 & -i \\ i & 0
\end{bmatrix},\;
Z = 
\begin{bmatrix}
     1 & 0 \\ 0 & -1
\end{bmatrix}.
$$
Note that because we are concerned with Hermitian qubit Pauli operators, we will only make use of $\pm 1$ multiples of $m$-qubit Pauli operators in $\mathbf{\mathcal{P}_m}$. 
We call an $m$-qubit operator a `cycle Pauli', denoted as $P_i$, if it realizes the $ith$ node of the cycle graph under discussion. We call an $m$-qubit operator an `edge Pauli', denoted as $L_i$, if it is a product of two neighbouring cycle Pauli operators, i.e. $L_i = P_iP_{i \oplus 1}$. Furthermore, throughout this work, we keep making use of the following properties of Pauli operators in $\mathbf{\mathcal{P}_m}$: (i) self-Inverse $P^2 = I$, $\forall \; P \in \mathcal{P}_m$ (ii) The product of two Hermitian Pauli operators is a Hermitian Pauli operator, i.e.~for $P,Q$ $\in \mathcal{P}_m$  if $P^{\dagger} = P$, $Q^{\dagger} = Q$, and $[P,Q] = 0$, then $(PQ)^{\dagger} = PQ$ where $PQ \in  \mathcal{P}_m.$ (iii) two Paulis either commute or anti-commute, i.e. either $[P,Q] = 0$ or $\{P,Q\} = 0$. We now describe the structural constraints on Pauli operators realizing a cycle scenario faithfully.

We divide up the constraints into two types (i) Forbidden edge Paulis (ii) Commutativity constraints. For a given $m$-qubit Pauli realization for a cycle scenario, the first constraint specifies the Pauli operators that can never be edge Pauli operators. While the second constraint defines the commutativity relations  within the set of edge Paulis and the ones across the edge and cycle Pauli operators.
\subsection{Forbidden Edge Paulis}
The following two are the forbidden properties of edge Pauli operators, given a faithful Pauli realization of a cycle scenario:
\begin{itemize}
    \item An edge Pauli $L_i=P_{i}P_{i \oplus 1}$ cannot be a multiple of any of the cycle Paulis constituting it, i.e. $L_i \neq k P_{i}$. This is because if $L_i  = k P_{i} $, then $ P_{i\oplus1}= kI $. This would mean that $P_{i\oplus1}$ will have to commute with all cycle Paulis in the cycle, which is forbidden because of the faithfulness condition.
    \item An edge Pauli $L_i=P_{i}P_{i \oplus 1}$ cannot be a multiple of any of the cycle Paulis not constituting it, i.e. $L_i \neq kP_l $ where $l \in \{0,1,2,3...,n-1\} \setminus \{i$, $ i\oplus 1\}$. This is because $L_i = k P_l $ along with the fact that $[L_i,P_i] = 0$ implies that $[P_i,P_l]=0$. In other words, the cycle Pauli $P_i$ must commute with a non-neighbouring Pauli $P_l$. This is forbidden because of the faithfulness condition.
    
\end{itemize} 
\subsection{Commutativity relations}
\label{commutativity}
The following commutativity relations prove to be very fruitful for the results to follow later. First three conditions define the commutative properties within the set of edge Pauli operators while the last one defines the commutation relation between edge and cycle Pauli operators: 
\begin{enumerate}
    \item No edge Pauli commutes with its nearest neighboring edge Paulis, i.e.~$\{L_i,L_{i\oplus 1}\}= 0$ and $\{L_i,L_{i\oplus (n-1)}\}= 0$.
    \item No edge Pauli commutes with its next-nearest edge Pauli operator except for when $n = 4$, i.e.~for all $n > 5$ cycle Pauli realizations $\{L_i,L_{i\oplus 2}\}= 0$ and $\{L_i,L_{i \oplus (n-2)}\}= 0$.
    \item For  $n>5$ , $[L_{i},L_j]=0$ only when $ i \oplus 3 \leq j \leq i\oplus (n-3) $. In words, an edge Pauli commutes with all edge Paulis except with its neighbouring and next-nearest neighbouring ones. 
    \item An edge Pauli commutes with all cycle Paulis except for the neighbours of the cycle Paulis constituting it, i.e.~$[L_i,P_l]=0$ only when $ l \neq i\oplus (n-1)$ and $i\oplus 2$. 
    
\end{enumerate} 
With these results we are now well equipped to begin our analysis involving faithful Pauli realizations of cycle scenarios. 
\section{Impossible faithful realizations}
\label{Impossible realizations}
Using the above-mentioned constraints, we can derive upper bounds on the size of the cycle scenario that is faithfully realizable by $m$-qubit Pauli operators. This question is motivated by our computational searches where we ran subgraph isomorphism tests to see which cycles are realizable by 2,3-qubit Pauli operators. We found that all cycles from size 4 to 6 are realizable by 2-qubit Paulis, whereas for 3-qubit Paulis all cycles from size 4 to 9, except, surprisingly, the 8-cycle are realizable. We now explain as to why a 7-cycle (or bigger) realization by 2-qubit Pauli operators is impossible. Fig.~\ref{unrealizability_of_7-cycle} will serve as an illustration for this proof.
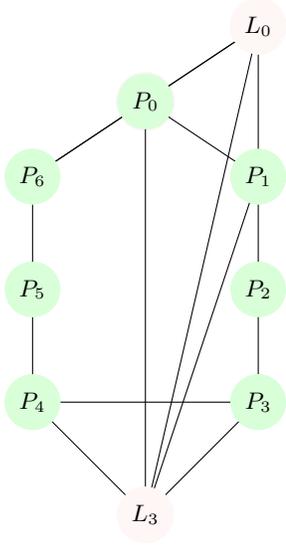
\begin{figure}
\begin{center}
\begin{tikzpicture}
[roundnode1/.style={circle, fill=green!15, minimum size=5mm}, roundnode2/.style={circle, fill=pink!15, minimum size=5mm}]
\node[roundnode1]{$P_0$};
\node[roundnode1]{$P_1$};
\node[roundnode1]{$P_2$};
\node[roundnode1]{$P_3$};
\node[roundnode1]{$P_4$};
\node[roundnode1]{$P_5$};
\node[roundnode1]{$P_6$};
\node[roundnode2]{$X_0$};
\node[roundnode2]{$X_4$};
\draw (0,0)node[roundnode1]{$P_0$}--(1.5,-1)node[roundnode1]{$P_1$};
\draw (1.5,-1)node[roundnode1]{$P_1$}--(1.5,-2.5)node[roundnode1]{$P_2$};
\draw (1.5,-2.5)node[roundnode1]{$P_2$}--(1.5,-4.0)node[roundnode1]{$P_3$};
\draw (1.5,-4.0)node[roundnode1]{$P_3$}--(-1.5,-4.0)node[roundnode1]{$P_4$};
\draw (-1.5,-4.0)node[roundnode1]{$P_4$}--(-1.5,-2.5)node[roundnode1]{$P_5$};
\draw (-1.5,-2.5)node[roundnode1]{$P_5$}--(-1.5,-1)node[roundnode1]{$P_6$};
\draw (-1.5,-1)node[roundnode1]{$P_6$}--(0,0)node[roundnode1]{$P_0$};
\draw (-1.5,-1)node[roundnode1]{$P_6$}--(0,0)node[roundnode1]{$P_0$};
\draw (0,0)node[roundnode1]{$P_0$}--(1.5,1.0)node[roundnode2]{$L_0$};
\draw (0,0)node[roundnode1]{$P_0$}--(1.5,1.0)node[roundnode2]{$L_0$};
\draw (1.5,-1)node[roundnode1]{$P_1$}--(1.5,1.0)node[roundnode2]{$L_0$};
\draw (1.5,-4.0)node[roundnode1]{$P_3$}--(0,-5.5)node[roundnode2]{$L_3$};
\draw (-1.5,-4.0)node[roundnode1]{$P_4$}--(0,-5.5)node[roundnode2]{$L_3$};
\draw (0,0)node[roundnode1]{$P_0$}--(0,-5.5)node[roundnode2]{$L_3$};
\draw (1.5,1.0)node[roundnode2]{$L_0$}--(0,-5.5)node[roundnode2]{$L_3$};
\draw (1.5,-1)node[roundnode1]{$P_1$}--(0,-5.5)node[roundnode2]{$L_3$};
\end{tikzpicture}
\caption{In a 2-qubit Pauli realization of a 7-cycle (in green) we introduce two edge Paulis (in pink) $L_0=P_{0}P_{1}$ ; $L_3=P_{3}P_{4}$ and some more edges allowed by the four conditions obtained in Section~\ref{commutativity}.}
\label{unrealizability_of_7-cycle}
\end{center}
\end{figure}

Assume that the cycle Paulis in Fig.~\ref{unrealizability_of_7-cycle} are 2-qubit Pauli operators. We now use the fact that a maximal stabilizer subgroup for $m$ qubits has $m$ independent generators to show that such a cycle is impossible to realize. For that consider two independent commuting Paulis \{$P_{0}$, $L_{0}$\} and the maximal subgroup, say $\mathcal{R}$, generated by them i.e. $ \mathcal{R} \equiv \langle P_{0}$,$L_{0} \rangle $. Since $L_3$ commutes with the generators in this maximal subgroup, hence $L_3 \in \mathcal{R}$ i.e. $L_3={P_{0}}^{\alpha}{L_{0}}^{\beta} $ where $\alpha,\beta \in \{0,1\}$. Consider the following possibilities:
\begin{itemize}
\item $\alpha = 0, \beta = 1$ $\implies$ $L_3=L_{0}$: This is forbidden because $[L_3,L_4] \neq 0$, but $[L_{0},L_{4}] = 0$. These relations follow from conditions 2 and 3 respectively in Sec.~\ref{commutativity}.
\item $\alpha = 1, \beta = 0\implies L_3=P_0$: This is forbidden because if it is true, then $P_1=I$ which breaks the faithfulness of the realization.
\item $\alpha = 0, \beta = 0 \implies L_3=I$: This is forbidden because if it is true, then $P_3=P_4$ which breaks the faithfulness of the realization.
\item $\alpha = 1, \beta = 1 \implies L_3=P_{0}L_{0} =P_{0}P_{0}P_{1}=P_1$: This is forbidden since it enforces the edges $(P_3,P_1)$ and $(P_4,P_1)$, hence violating the faithfulness of the realization.
\end{itemize}
This exhausts all the possibilities and hence such a cycle realization is prohibited. It can be trivially seen that the same argument will hold for any two-qubit Pauli realization of an $n$-cycle, $\forall$  $n \geq 7$. We now generalise this result in the form of the following theorem.
\begin{theorem}
\label{no_go_theorem}
    For a given $m$, there exists no faithful $m$-qubit Pauli realization of an $n$-cycle scenario such that $n > 3m$. 
\end{theorem}
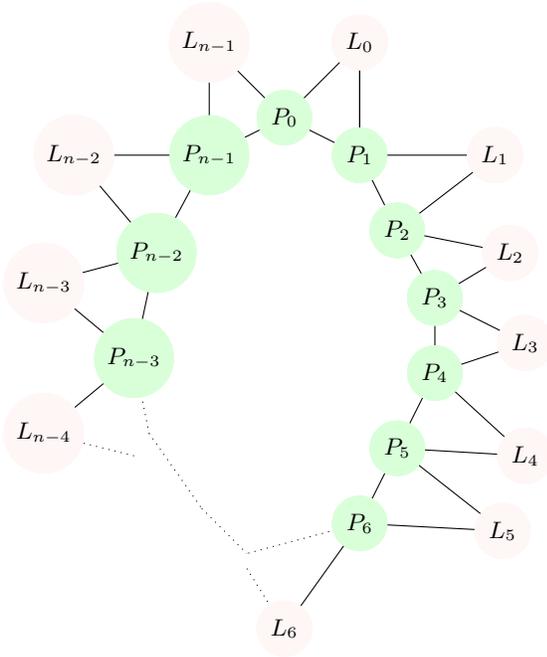
\begin{figure}
\begin{center}
\begin{tikzpicture}
[roundnode1/.style={circle, fill=green!15, minimum size=3mm}, roundnode2/.style={circle, fill=pink!15, minimum size=3mm}]
\draw (0,0)node[roundnode1]{$P_0$}--(1.0,-0.5)node[roundnode1]{$P_1$};
\draw (1.0,-0.5)node[roundnode1]{$P_1$}--(1.5,-1.5)node[roundnode1]{$P_2$};
\draw (1.5,-1.5)node[roundnode1]{$P_2$}--(2.0,-2.4)node[roundnode1]{$P_3$};
\draw (2.0,-2.4)node[roundnode1]{$P_3$}--(2.0,-3.4)node[roundnode1]{$P_4$};
\draw (2.0,-3.4)node[roundnode1]{$P_4$}--(1.5,-4.4)node[roundnode1]{$P_5$};
\draw (1.5,-4.4)node[roundnode1]{$P_5$}--(1.0,-5.4)node[roundnode1]{$P_6$};
\draw [dotted](1.0,-5.4)--(-0.5,-5.8);
\draw (0,0)node[roundnode1]{$P_0$}--(-1.0,-0.5)node[roundnode1]{$P_{n-1}$};
\draw (-1.0,-0.5)node[roundnode1]{$P_{n-1}$}--(-1.7,-1.8)node[roundnode1]{$P_{n-2}$};
\draw (-1.7,-1.8)node[roundnode1]{$P_{n-2}$}--(-2.0,-3.2)node[roundnode1]{$P_{n-3}$};
\draw [dotted](-2.0,-3.2)--(-1.8,-4.2);
\draw [dotted](-1.8,-4.2)--(-1.1,-5.2);
\draw [dotted](-1.1,-5.2)--(-0.5,-5.8);
\draw (0,0)node[roundnode1]{$P_0$}--(-1.0,1.0)node[roundnode2]{$L_{n-1}$};
\draw (-1.0,-0.5)node[roundnode1]{$P_{n-1}$}--(-1.0,1.0)node[roundnode2]{$L_{n-1}$};
\draw (-1.0,-0.5)node[roundnode1]{$P_{n-1}$}--(-2.8,-0.5)node[roundnode2]{$L_{n-2}$};
\draw (-1.7,-1.8)node[roundnode1]{$P_{n-2}$}--(-2.8,-0.5)node[roundnode2]{$L_{n-2}$};
\draw (-1.7,-1.8)node[roundnode1]{$P_{n-2}$}--(-3.2,-2.2)node[roundnode2]{$L_{n-3}$};
\draw (-2.0,-3.2)node[roundnode1]{$P_{n-3}$}--(-3.2,-2.2)node[roundnode2]{$L_{n-3}$};
\draw (-2.0,-3.2)node[roundnode1]{$P_{n-3}$}--(-3.2,-4.2)node[roundnode2]{$L_{n-4}$};
\draw [dotted](-2.0,-4.5)--(-3.2,-4.2)node[roundnode2]{$L_{n-4}$};
\draw (0,0)node[roundnode1]{$P_0$}--(1,1)node[roundnode2]{$L_0$};
\draw (1.0,-0.5)node[roundnode1]{$P_1$}--(1,1)node[roundnode2]{$L_0$};
\draw (1.0,-0.5)node[roundnode1]{$P_1$}--(1,1)node[roundnode2]{$L_0$};
\draw (1.0,-0.5)node[roundnode1]{$P_1$}--(2.8,-0.5)node[roundnode2]{$L_1$};
\draw (1.5,-1.5)node[roundnode1]{$P_2$}--(2.8,-0.5)node[roundnode2]{$L_1$};
\draw (1.5,-1.5)node[roundnode1]{$P_2$}--(3.0,-1.8)node[roundnode2]{$L_2$};
\draw (2.0,-2.4)node[roundnode1]{$P_3$}--(3.0,-1.8)node[roundnode2]{$L_2$};
\draw (2.0,-2.4)node[roundnode1]{$P_3$}--(3.2,-3.0)node[roundnode2]{$L_3$};
\draw (2.0,-3.4)node[roundnode1]{$P_4$}--(3.2,-3.0)node[roundnode2]{$L_3$};
\draw (2.0,-3.4)node[roundnode1]{$P_4$}--(3.2,-4.5)node[roundnode2]{$L_4$};
\draw (1.5,-4.4)node[roundnode1]{$P_5$}--(3.2,-4.5)node[roundnode2]{$L_4$};
\draw (1.5,-4.4)node[roundnode1]{$P_5$}--(2.9,-5.5)node[roundnode2]{$L_5$};
\draw (1.0,-5.4)node[roundnode1]{$P_6$}--(2.9,-5.5)node[roundnode2]{$L_5$};
\draw (1.0,-5.4)node[roundnode1]{$P_6$}--(0.0,-6.8)node[roundnode2]{$L_6$};
\draw [dotted](-0.5,-6)--(0.0,-6.8)node[roundnode2]{$L_6$};
\end{tikzpicture}
\caption{An arbitrary, faithful $m$-qubit Pauli realization of an $n$-cycle (green) with the corresponding edge Paulis $L_i=P_{i}P_{i\oplus 1}$ (pink), the dotted lines refers to the Paulis (including edge Paulis) not drawn but are part of the cycle.} 
\label{generic_cycle}
\end{center}
\end{figure}
\begin{proof}
Consider an arbitrary $m$-qubit Pauli realization of an $n$-cycle as shown in Fig.~\ref{generic_cycle}, including the edge Pauli operators. Now, collect together the following Pauli operators: 
$$S \equiv\{P_0,L_0,L_3,L_6,\ldots,L_{n-3}\}$$. By the constraints shown in Section~\ref{commutativity}, all Paulis in $S$ pairwise commute. More precisely, this is because:
\begin{enumerate}
\item The edge Paulis in $S$ all pairwise commute due to condition 3 in Section~\ref{commutativity}.
\item  Each $L_j \in S$ commutes with $P_0$ via the last condition highlighted in Section~\ref{commutativity}.
\end{enumerate}
We now prove that a certain subset of the Pauli operators in $S$ are also independent \footnote{independent means that no Pauli operator can be written as a product of other commuting Paulis}:
First pick all the elements up to $L_i$ starting from $P_0$ from the realization (in the clockwise sense of Fig.~\ref{generic_cycle}) from the set $S$ i.e. $S_i \equiv \{P_0,L_0,L_3,L_6,...L_{i-3},L_i\}$. Consider the edge Pauli $L_{i+1}$. Via conditions 3 and 4 in Section~\ref{commutativity}- $L_{i+1}$ commutes with every element in $S_{i-1}=\{P_0,L_0,L_3,L_6,...,L_{i-3}\}$. By construction $L_i$ also commutes with every element in $S_{i-1}$. Additionally, via condition 1 in Section~\ref{commutativity}, $[L_i,L_{i+1}]\neq 0$. Therefore, $L_i \notin \langle P_0,L_0,L_3,L_6,...,L_{i-3} \rangle$. This is because, if it did then $[L_i,L_{i+1}]= 0$ which we know can not hold true (via condition 1 of Section~\ref{commutativity}). The above argument holds for any $L_i \in S $, except for the last Pauli operator i.e. $L_{n-3}$ . For this case, the Pauli $L_{(n-3)+1}$ does not commute with $P_0,L_0$ due to condition 4 in Section~\ref{commutativity} and hence $L_{n-3}$ cannot be guaranteed to be independent of the rest of the elements in $S$. So, we throw away this element from $S$. This proves that all Paulis in $S \setminus \{L_{n-3}\}$ are independent. We now update $S$: 
$$S = \{P_0,L_0,L_3,L_6,\ldots,L_{3(\lceil\frac{n}{3}\rceil-2)}\}$$
We know that for any given $m$, any maximal subgroup is generated by $m$ independent, pairwise commuting Paulis. Therefore, $|S| \leq m$. We can further write, $S =\{P_0\} \cup \{L_{3j}\}_{j=0}^{\lceil\frac{n}{3}\rceil-2}$. The case where $|S|= m$ means that $\lceil\frac{n}{3}\rceil \leq m$, hence $n \leq 3m$. 

\end{proof}
Having obtained $3m$ as the upper bound, we further ran subgraph isomorphism tests to test the tightness of this bound. Unlike for $m = 2,3$, when $m = 4$, $3m = 12$ is not a tight upper bound. For this case, all cycles from size 4 to 9 are faithfully realizable. Similarly, when $m = 5$, $3m = 15$ is not a tight upper bound, but all cycles from size 4 to 12 are realizable. 
\section{Possible faithful realizations}
\label{Actualizing faithful realizations}
We now focus our attention on some constructions that provide faithful realizations of Pauli operators over cycle scenarios. We will now provide a few constructions 
that provide faithful realizations, in accordance with the no-go theorem. But before doing so, we highlight a trivial but useful fact via the following proposition.
\begin{proposition}
\label{proposition}
    If an $n$-cycle is realizable by $m$-qubit Pauli operators, it is also realizable by ($m+1$)-qubit Pauli operators.
\end{proposition}
\begin{proof}
    Given a $m$-qubit Pauli realization of an $n$-cycle, one can append a new qubit to each such Pauli operator and assign operator $I$ or the same Pauli operator corresponding to the added operator. This produces a $m+1$-qubit Pauli realization for the $n$-cycle.
\end{proof}
\subsection{$m+2$ cycle with $m$-qubit Paulis}
\label{m+2 construction}
Before presenting our construction, we highlight a relevant construction from Ref.~\cite{abramsky2024commutation}. In this work, the authors provided a general construction to realize any compatibility graph via qudit Pauli operators. Their construction applied to the cycle scenarios and qubit Pauli operators guarantees a realization of $m$-cycle using $m$ qubit Pauli operators. Below we present the construction in Ref.~\cite{abramsky2024commutation}, to then provide ours which is an improvement on it: 
\begin{equation}
\begin{split}
XIII\ldots III\\
IXII\ldots III\\
ZIXI\ldots III\\
ZZIX\ldots III\\
ZZZI\ldots III\\
\ldots\\
ZZZZ\ldots IXI\\
IZZZ\ldots ZIX\\
\end{split}
\end{equation}
Here, starting from the first $m$-qubit Pauli $XIII\ldots III$, and then shifting the first single qubit Pauli operator $X$ to the second qubit, we keep adding a new single qubit $Z$ operator from the left, at the first qubit, while fixing the first $m-1$ Pauli operators of the previous $m$-qubit Pauli to its right. This ensures that a given Pauli is only compatible with the Pauli operator right before it, while being incompatible with the rest. Following this strategy one reaches a point with a path realization of $m-1$, $m$-qubit Pauli operators with the last Pauli being $ZZZZ\ldots IXI$. At this point, the pattern of appending single qubit $Z$ Pauli operator from the left is broken, and in the final step, a single qubit $I$ is introduced at the left while using the first $m-1$ Pauli operators of the previous $m$-qubit Pauli to its right. This ensures that this $m$-qubit Pauli operator is compatible with the previous one and the first $m$-qubit Pauli operators only. Overall, this construction from Ref.~\cite{abramsky2024commutation} guarantees an $m$-qubit Pauli realization of $m$-cycle. For a visual proof, one can track down how the single qubit $X$ operator in the construction shifts, from first qubit to the $m$th qubit, while the $Z$ operator is appended except at the last Pauli.

Our improvement on this construction gurantees an $m+2$ sized cycle using $m$-qubit Pauli operators. Basically, we intervene at the final step of the previous construction, where instead of appending a single qubit $I$ operator at the first qubit, we append a single qubit $Z$ operator at the first qubit to attain $ZZZZ\ldots ZIX$ and then do the same to attain $ZZZZ\ldots ZZI$. Now, we apply the last step in the previous construction and obtain $IZZZ\ldots ZZZ$, an $m$-qubit Pauli operator that commutes with the first and the Pauli operator right before it while anti-commuting with the rest. Overall, the construction is as follows:
\begin{equation}
\begin{split}
XIII\ldots III\\
IXII\ldots III\\
ZIXI\ldots III\\
ZZIX\ldots III\\
ZZZI\ldots III\\
\ldots \\
ZZZZ\ldots IXI\\
ZZZZ\ldots ZIX\\
ZZZZ\ldots ZZI\\
IZZZ\ldots ZZZ
\end{split}
\end{equation}
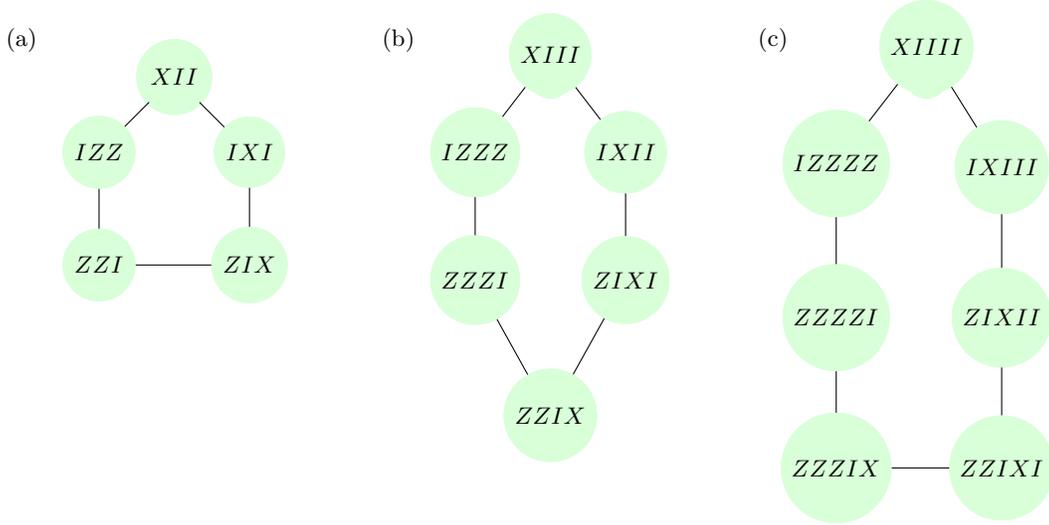
\begin{figure*}
\begin{center}
\begin{tikzpicture}
[roundnode1/.style={circle, fill=green!15, minimum size=05mm}, roundnode2/.style={circle, fill=pink!15, minimum size=5mm}]
\node[text width=0.4cm] at (-2,0.5) 
    {(a)};
\begin{scope}
\node[roundnode1]{$XII$};
\node[roundnode1]{$IXI$};
\node[roundnode1]{$ZIX$};
\node[roundnode1]{$ZZI$};
\node[roundnode1]{$IZZ$};
\node[roundnode1]{$P_5$};
\node[roundnode1]{$P_6$};
\draw (0,0)node[roundnode1]{$XII$}--(1.0,-1)node[roundnode1]{$IXI$};
\draw (1.0,-1)node[roundnode1]{$IXI$}--(1,-2.5)node[roundnode1]{$ZIX$};
\draw (1,-2.5)node[roundnode1]{$ZIX$}--(-1,-2.5)node[roundnode1]{$ZZI$};
\draw (-1,-2.5)node[roundnode1]{$ZZI$}--(-1,-1)node[roundnode1]{$IZZ$};
\draw (-1,-1)node[roundnode1]{$IZZ$}--
(0,0)node[roundnode1]{$XII$};
\end{scope}

\begin{scope}[xshift=5cm]
\node[text width=0.4cm] at (-2,0.5) 
    {(b)};
\node[roundnode1]{$1$};
\node[roundnode1]{$2$};
\node[roundnode1]{$3$};
\node[roundnode1]{$4$};
\node[roundnode1]{$5$};
\node[roundnode1]{$6$};
\draw (0,0.3)node[roundnode1]{$XIII$}--(1.0,-1)node[roundnode1]{$IXII$};
\draw (1.0,-1)node[roundnode1]{$IXII$}--(1,-2.7)node[roundnode1]{$ZIXI$};
\draw (1,-2.7)node[roundnode1]{$ZIXI$}--(0,-4.5)node[roundnode1]{$ZZIX$};
\draw (-0,-4.5)node[roundnode1]{$ZZIX$}--(-1,-2.7)node[roundnode1]{$ZZZI$};
\draw (-1,-2.7)node[roundnode1]{$ZZZI$}--
(-1,-1)node[roundnode1]{$IZZZ$};
\draw (-0,0.3)node[roundnode1]{$XIII$}--
(-1,-1)node[roundnode1]{$IZZZ$};
\end{scope}
\begin{scope}[xshift=10cm]
\node[text width=0.4cm] at (-2,0.5) 
    {(c)};
\node[roundnode1]{$1$};
\node[roundnode1]{$2$};
\node[roundnode1]{$3$};
\node[roundnode1]{$4$};
\node[roundnode1]{$5$};
\node[roundnode1]{$6$};
\node[roundnode1]{$7$};
\draw (0,0.4)node[roundnode1]{$XIIII$}--(1.0,-1.2)node[roundnode1]{$IXIII$};
\draw (1.0,-1.2)node[roundnode1]{$IXIII$}--(1,-3.2)node[roundnode1]{$ZIXII$};
\draw (1,-3.2)node[roundnode1]{$ZIXII$}--(1,-5.2)node[roundnode1]{$ZZIXI$};
\draw (1,-5.2)node[roundnode1]{$ZZIXI$}--(-1.2,-5.2)node[roundnode1]{$ZZZIX$};
\draw (-1.2,-5.2)node[roundnode1]{$ZZZIX$}--
(-1.2,-3.2)node[roundnode1]{$ZZZZI$};
\draw (-1.2,-3.2)node[roundnode1]{$ZZZZI$}--
(-1.2,-1.15)node[roundnode1]{$IZZZZ$};
\draw (0,0.4)node[roundnode1]{$XIIII$}--
(-1.2,-1.15)node[roundnode1]{$IZZZZ$};
\end{scope}
\end{tikzpicture}
\caption{Realization of $m+2$ sized cycle by $m$-qubit Paulis (a) $m = 3$ (b) $m = 4$ (c) $m = 5$.} 
\label{m+2_with_m}
\end{center}
\end{figure*}
Fig.~\ref{m+2_with_m} shows specific examples of this construction. Inductively, Proposition~\ref{proposition} and this construction guarantees faithful realizations of all cycles up to size $m+2$ by $m$-qubit Pauli operators. This is because, by the above construction, single ($m = 1$) qubit Paulis can realize a 3-cycle, two ($m = 2$) qubit Pauli operators can realize a 4-cycle, similarly three ($m = 3$) qubit Pauli operators realize a 5-cycle and so on, and Proposition~\ref{proposition} guarantees that $m$-qubit Pauli operators realize all cycles realizable by $k$-qubits $\forall\; 1 \leq k \leq m-1$. 

\subsection{Converting Path graph realizations to cycle realizations}
Before presenting our next construction with which we almost double the number of cycles we get compared to the previous constructions, we need a few results about faithful Pauli realizations of path graphs and their conversion to faithful cycle realizations. We capture these results in the form of theorems and corollaries.
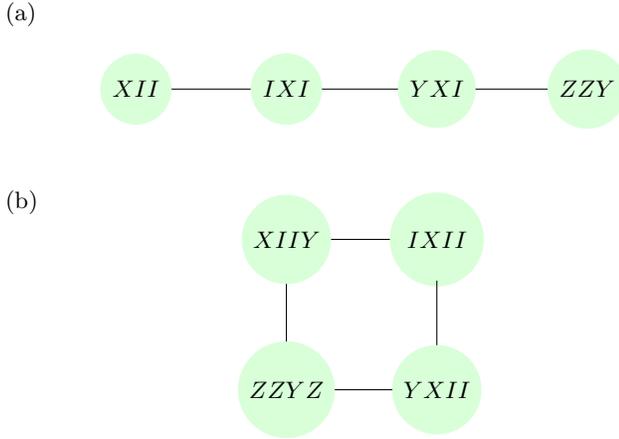
\begin{figure}
\begin{center}
\begin{tikzpicture}
[roundnode1/.style={circle, fill=green!15, minimum size=05mm}, roundnode2/.style={circle, fill=pink!15, minimum size=5mm}]
\node[text width=0.4cm] at (-11.5,1) 
    {(a)};
\node[text width=0.4cm] at (-11.5,-1.5) 
    {(b)};
\begin{scope}
\draw (-10,0)node[roundnode1]{$XII$}--(-8,0)node[roundnode1]{$IXI$};
\draw (-8,0)node[roundnode1]{$IXI$}--(-6,0)node[roundnode1]{$YXI$};
\draw (-6,0)node[roundnode1]{$YXI$}--(-4,0)node[roundnode1]{$ZZY$};
\end{scope}
\begin{scope}[yshift=-2cm]
\draw (-8,0)node[roundnode1]{$XIIY$}--(-6,0)node[roundnode1]{$IXIII$};
\draw (-6,0)node[roundnode1]{$IXII$}--(-6,-2)node[roundnode1]{$YXII$};
\draw (-6,-2)node[roundnode1]{$YXII$}--(-8,-2)node[roundnode1]{$ZZYZ$};
\draw (-8,-2)node[roundnode1]{$ZZYZ$}--(-8,0)node[roundnode1]{$XIIY$};
\end{scope}
\end{tikzpicture}
\caption{Converting a path graph into a cycle by addition of a qubit (a) a 3-qubit path graph realization of $H_4$ (b) Adding a qubit and assigning anti-commuting Pauli operators, $Y$ and $Z$, to the first and last nodes and $I$ to the rest to get a 4-qubit realization of $C_4$.} 
\label{path_to_cycle_2_qubits}
\end{center}
\end{figure}
\begin{theorem}
\label{Theorem2}
If a path graph with $l$ nodes is faithfully realizable by $m$-qubit Pauli operators, then one can always extend it to a faithful realisation of $l$-cycle using $(m+1)$-qubit Pauli operators. 
\end{theorem}
\begin{proof}
    Let $H_l$ be a path graph with $l$ nodes where each node is some $m$-qubit Pauli operator. Furthermore, let the $m$-qubit Pauli operator on the $i$th node of graph $H_l$ be denoted as $P_i$. Since it is a faithful realization, Paulis $P_1$ and $P_l$ anti-commute. Recall that this means there are odd number of qubits corresponding to which Pauli operators across $P_1$ and $P_l$ anticommute. Consider the following $l$, ($m+1$)-qubit Pauli operators:  $XP_1, IP_2,\ldots,IP_{l-1},YP_l$. In words, corresponding to the ($m+1$)th qubit we introduce Pauli $X$ for $P_1$, Pauli $Y$ for $P_l$ while appending $I$ for all $P_i$ where $i \in [2,l-1]$. This leads to a $l$-cycle faithful realization by these Pauli operators because it preserves all commutation and anti-commutation relations as in $H_l$ except between $XP_1$ and $YP_l$ which now commute due to even number of anti-commutating single qubit Pauli operators across these two $(m+1)$-qubit Pauli operators, see Fig.~\ref{path_to_cycle_2_qubits} for an example.
\end{proof}
\begin{corollary}
\label{corollary3}
    If a path graph with $l$ nodes is faithfully realizable by $m$-qubit Pauli operators, then one can always construct a faithful realization of $k$-cycle using $(m+1)$-qubit Pauli operators, $\forall \; 3 \leq k \leq l$. 
\end{corollary}
\begin{proof}
Any $m$-qubit faithful realization of a path graph with $l$ nodes contains a faithful $m$-qubit realization of each of its sub-paths 
with $3$ to $l$ nodes. Applying Theorem~\ref{Theorem2} to each such sub path proves our claim. 
\end{proof}
In Section~\ref{Impossible realizations}, we mentioned that for $m = 3$, one can faithfully realize all cycles from size 4 to 9 except the 8-cycle. Via proposition~\ref{proposition}, 4-qubit Pauli operators realize all the cycles realised by 3-qubit Pauli operators. Additionally, we guarantee that 4-qubit Pauli operators also faithfully realize the 8-cycle. This is because
the 3-qubit Pauli realization corresponding to $C_9$ has an induced path graph with 8 nodes, which via Theorem~\ref{Theorem2} guarantees our claim. More generally, Theorem~\ref{Theorem2} guarantees that if the maximum size of the faithfully realized cycle by a given $m$-qubit Pauli operators is $k$ ($\leq 3m$), then all cycles from size $3$ to $k$ will be necessarily faithfully realized by $(m+1)$-qubit Pauli operators. 

We now present a result that gives us the power to combine two faithful Pauli realizations of path graphs. We capture this precisely in the following theorem.
\begin{figure*}
\begin{center}
\begin{tikzpicture}
[roundnode1/.style={circle, fill=green!15, minimum size=05mm}, roundnode2/.style={circle, fill=pink!15, minimum size=5mm}]
\node[text width=6cm] at (3.5,0) 
    {$\rightarrow$ $m$-qubit Pauli realization of $H_l$ ($S_1$)};
\node[text width=6cm] at (-5.2,-1.5) 
    {$m'$-qubit Pauli realization of $H_l'$ ($S_2$) $\leftarrow$};
\draw [decorate,decoration={brace,amplitude=5pt,raise=4ex}]
  (-0.5,-3.5) -- (-7.5,-3.5);
\node[text width=10cm] at (-2.2,-4.8) 
    {$(m+m')$-qubit Pauli realization of $H_{l+l'-2}$}; 
\begin{scope}
\draw (-10,0)node[roundnode1]{$P_1$}--(-8,0)node[roundnode1]{$P_2$};
\draw (-8,0)node[roundnode1]{$P_2$}--(-6,0)node[roundnode1]{$---$};
\draw (-6,0)node[roundnode1]{$---$}--(-4,0)node[roundnode1]{$P_{l-2}$};
\draw (-4,0)node[roundnode1]{$P_{l-2}$}--(-2,0)node[roundnode1]{$P_{l-1}$};
\draw (-2,0)node[roundnode1]{$P_{l-1}$}--(0,0)node[roundnode1]{$P_{l}$};
\end{scope}
\begin{scope}
\draw (-2,-1.5)node[roundnode1]{$Q_1$}--
(0,-1.5)node[roundnode1]{$Q_2$};
\draw (0,-1.5)node[roundnode1]{$Q_2$}--
(2,-1.5)node[roundnode1]{$---$};
\draw (2,-1.5)node[roundnode1]{$---$}--
(4,-1.5)node[roundnode1]{$Q_{l'-1}$};
\draw (4,-1.5)node[roundnode1]{$Q_{l'-1}$}--
(6,-1.5)node[roundnode1]{$Q_{l'}$};
\end{scope}
\begin{scope}
\draw (-10,-3)node[roundnode1]{$P_1I$}--(-8,-3)node[roundnode1]{$P_2I$};
\draw (-8,-3)node[roundnode1]{$P_2I$}--(-6,-3)node[roundnode1]{$---$};
\draw (-6,-3)node[roundnode1]{$---$}--(-4,-3)node[roundnode1]{$P_{l-2}I$};
\draw (-4,-3)node[roundnode1]{$P_{l-2}I$}--(-2,-3)node[roundnode1]{$P_{l-1}Q_1$};
\draw (-2,-3)node[roundnode1]{$P_{l-1}Q_1$}--
(0,-3)node[roundnode1]{$P_lQ_2$};
\draw (0,-3)node[roundnode1]{$P_lQ_2$}--
(2,-3)node[roundnode1]{$---$};
\draw (2,-3)node[roundnode1]{$---$}--
(4,-3)node[roundnode1]{$P_lQ_{l-1}$};
\draw (4,-3)node[roundnode1]{$P_lQ_{l'-1}$}--
(6,-3)node[roundnode1]{$P_lQ_{l'}$};
\end{scope}
\end{tikzpicture}
\caption{Clubbing together a faithful $m$-qubit Pauli realization of path graph $H_l$ with a faithful $m'$-qubit Pauli realization of $H_{l'}$, to obtain a ($m+m'$)-qubit Pauli realization of $H_{l+l'-2}$. This works by introducing $m'$-qubit Identity operators in the first $l-2$ operators of the realization of path graph $H_l$ while introducing the Pauli realization for $H_{l'}$ from the $(l-1)$th node in $H_l$. In the additional $l'-2$ nodes, one then fixes $P_l$ in the first $m$ qubits.} 
\label{Path_graph_combination_2}
\end{center}
\end{figure*}
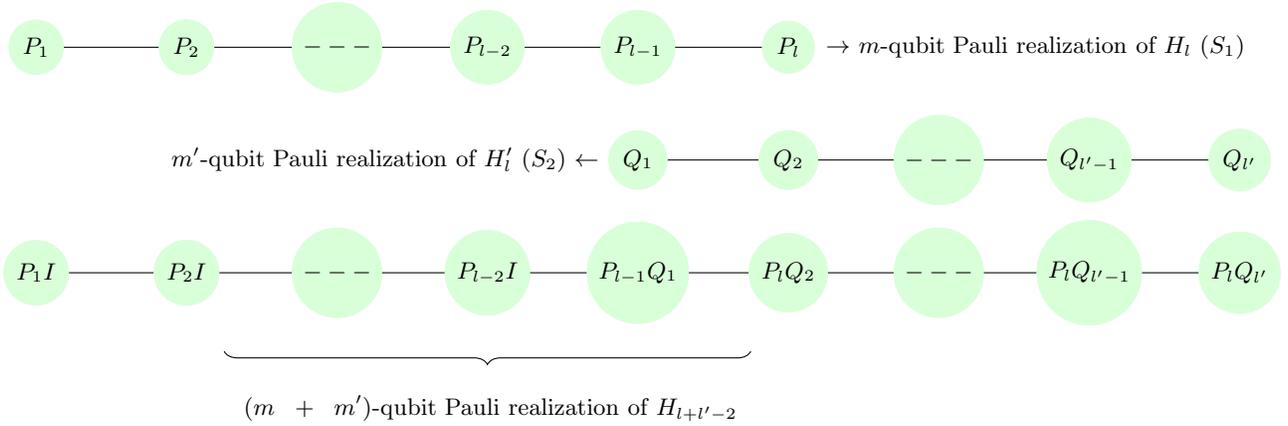 

\begin{theorem}
\label{Theorem4}
Given a faithful $m$-qubit Pauli realization of a path graph $H_l$, and a faithful $m'$-qubit Pauli realization of path graph $H_{l'}$, one can construct a faithful ($m+m'$)-qubit Pauli realization of path graph $H_{l+l'-2}$.
\end{theorem}
\begin{proof}
    The construction works as follows, see Fig.~\ref{Path_graph_combination_2} for illustration. Consider $P_1,P_2,\ldots,P_l$, in the same order, as the $m$-qubit Pauli realization of $H_l$ and $Q_{1},Q_{2},\ldots,Q_{l'}$, in the same order, as the $m'$-qubit Pauli realization of $H_{l'}$. To each of the first $l-2$, $m$-qubit Pauli operators in the path $H_l$, append a $m'$-qubit $I$ operator, while to the last two operators $P_{l-1}$ and $P_{l}$, append $Q_{1}$, $Q_{2}$ respectively. This construction, let's call it $S_1$, produces a $(m+m')$-qubit faithful realization of path graph $H_l$. Faithful because the appended $m'$-qubit Pauli operators are all pairwise compatible and hence don't perturb the compatiblity relations associated to the initial $m$-qubit Pauli realization of $H_l$. Now consider $l'-2$ copies of the $m$-qubit Pauli $P_l$ and append $Q_{3}, Q_{4}, \ldots, Q_{l'}$ respectively to each of the copies. This construction, let's call it $S_2$, produces a faithful $(m+m')$-qubit Pauli realization of $H_{l'-2}$. We can now stack $S_2$ underneath $S_1$ and show that this leads to a faithful ($m+m'$)-qubit Pauli realization of path graph $H_{l+l'-2}$:
    Anti-commutativity across each node in $S_2$ with the first $l-2$ nodes in $S_1$ is ensured because, for each such node in $S_1$, the $m$-qubit Pauli operator (corresponding to first $m$ qubits) anti-commutes with the corresponding operators of every node in $S_2$,
    while the remaining $m'$-qubit Pauli operator commutes with the corresponding operators of every node in $S_2$. On the other hand, for the last two nodes in $S_1$, the $m$-qubit Pauli operator corresponding to the first $m$ qubits commutes with the corresponding $m$-qubit Pauli operator for every node in $S_2$. This means that the commutativity relationship of the last two nodes in $S_1$ with all of the nodes in $S_2$ is decided by the last $m'$ qubits. All the Pauli operators corresponding to these qubits from $S_1$ to $S_2$ are $Q_{1},Q_{2},\ldots,Q_{l'}$ which is isomorphic to $H_{l'}$. This suggests that $S_2$ stacked underneath $S_1$ leads to a faithful ($m+m'$)-qubit Pauli realization of $H_{l+l'-2}$.
\end{proof}
\begin{corollary}
\label{corollary 5}
    Given a faithful $m$-qubit Pauli realization of a path graph $H_l$, there exists a faithful $(m+1)$-qubit Pauli realization of path graph $H_{l+1}$. 
\end{corollary}
\begin{proof}
    This follows from Theorem~\ref{Theorem4} where we select $m' = 1$ which gives a faithful realization of $H_{3}$ (i.e. $l' = 3$). One example of this realization is $X-I-Y$. With the faithful $m$-qubit Pauli realization of a path graph $H_l$, it will produce a faithful $(m+1)$-qubit Pauli realization of path graph $H_{l+1}$.  
\end{proof}
\begin{figure*}
\begin{center}
\begin{tikzpicture}
[roundnode1/.style={circle, fill=green!15, minimum size=05mm}, roundnode2/.style={circle, fill=pink!15, minimum size=5mm}]
\begin{scope}
\node[roundnode1]{$XII$};
\node[roundnode1]{$IXI$};
\node[roundnode1]{$ZIX$};
\node[roundnode1]{$ZZI$};
\node[roundnode1]{$IZZ$};
\node[roundnode1]{$P_5$};
\node[roundnode1]{$P_6$};
\draw (-10,0)node[roundnode1]{$XII$}--(-8,0)node[roundnode1]{$IXI$};
\draw (-8,0)node[roundnode1]{$IXI$}--(-6,0)node[roundnode1]{$YXI$};
\draw (-6,0)node[roundnode1]{$YXI$}--(-4,0)node[roundnode1]{$ZZY$};
\draw (-4,0)node[roundnode1]{$ZZY$}--(-2,0)node[roundnode1]{$YZX$};
\draw (-2,0)node[roundnode1]{$YZX$}--
(0,0)node[roundnode1]{$YZI$};
\draw (0,0)node[roundnode1]{$YZI$}--
(2,0)node[roundnode1]{$YZY$};
\draw (2,0)node[roundnode1]{$YZY$}--
(4,0)node[roundnode1]{$YYX$};
\end{scope}
\end{tikzpicture}
\caption{A faithful 3-qubit Pauli realization of a $H_8$.} 
\label{8-cycle_realization}
\end{center}
\end{figure*}
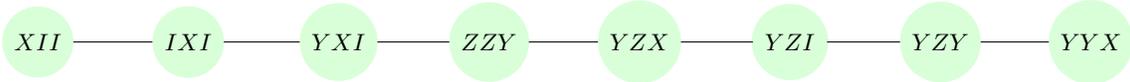
\subsection{$2m$ or a $2m-1$ sized cycle with $m$-qubit Paulis}
We are now ready to provide our construction. Given a faithful $m$-qubit Pauli realization of a path graph $H_l$ and repeated application of Theorem~\ref{Theorem4} for $p\;  (\geq 0)$ rounds, where we keep using the faithful $m'$-qubit Pauli realization of path graph $H_{l'}$ in each round, gives us a faithful $(m+pm')$-qubit Pauli realization of path graph $H_{l+p(l'-2)}$. This allows us to construct faithful realizations of cycles for arbitrary number of qubits. Consider the case when $m = m' = 3$ and $l = l'= 8$. Fig.~\ref{8-cycle_realization} shows a faithful realization of $H_8$ with three qubits. Starting from this realization, consider $p$ repeated applications of Theorem~\ref{Theorem4}. This provides a faithful $M$-qubit Pauli realization of path graph $H_{L}$, where $M = 3+3p$ and $L = 8 + 6p$. Based on this construction, we now define three cycle constructions as follows:
\begin{enumerate}
    \item Application of Theorem~\ref{Theorem2} to this construction gives us a $L$-cycle using $M' = M+1$ qubits. Note that $L = 2M'$. 
    \item Application of corollary~\ref{corollary 5} to this construction gives us a faithful $(M+1)$-qubit path realization of $H_{L+1}$. Further applying Theorem~\ref{Theorem2} gives us a faithful $M'$-qubit realization of ${(L+1)}$-cycle where $M' = M+2$. Note that $L + 1 = 2M'-1$.
    \item Application of Theorem~\ref{Theorem4} such that $m' = 2$, and $l'=2$, gives us a faithful $(M+2)$-qubit Pauli realization of path graph $H_{L+3}$. Further applying Theorem~\ref{Theorem2} gives us a faithful $M'$-qubit path realization of ${(L+3)}$-cycle where $M' = M+3$. Note that $L + 3 = 2M'-1$.
\end{enumerate} 
Note that since $M = 3(1+p)$, the cycle realization in construction (1) implies that with number of qubits $M' \equiv 1 $ (mod 3) we can faithfully realize a cycle of size $2M'$ while constructions (2) and (3) imply that with number of qubits $M' \equiv 0 $ (mod 3) or $M' \equiv 2 $ (mod 3), we can faithfully realize a cycle of size $2M'-1$. This means that a faithful Pauli realization of a cycle of size $2M'-1$ or $2M'$ is guaranteed with $M'$ qubits. Corollary~\ref{corollary3} further implies that with $M' + 1$ qubits all cycles from size 3 up to size $2M'-1$ or $2M'$ are faithfully realizable. We can further rephrase this result as follows. With $M'$ qubits all cycles from size 3 up to size $2M' - 3$ or $2M'-2$ can be faithfully realized by Pauli operators depending on whether $M \not\equiv 1$ (mod 3) or $M \equiv 1$ (mod 3) respectively. This leaves open the question of whether there exists a faithful realization of the cycle of size $2M' - 2$ when $M \not\equiv 1$ (mod 3) and $2M'-1$ when $M \equiv 1$ (mod 3). This finishes the descriptions of our constructions.

\section{Contextuality with Pauli cycles}\label{sec:contextuality_pauli}
We now focus on the contextuality witnessing aspects of the Pauli realizations of $n$-cycle scenarios.
Luckily, the commutation properties in Section~\ref{commutativity} turn out to be enough for this characterisation. The first step in every proof below takes inspiration from Tsirelson's original idea of squaring the operator to derive quantum bound for the CHSH inequality. We dedicate these proofs to his contribution to the field of quantum foundations.

For this analysis, we make use of the complete characterization of the non-contextual polytopes of $n$-cycle scenarios provided in Ref.~\cite{araujo2013all}:
$$ \sum_{i=0}^{n-1}\gamma_{i}\langle A_{i}A_{i\oplus 1}\rangle \overset{\textit{NCHV}}{\leq} n- 2 $$ where $\gamma_{i} \in \{-1,+1\}$ with $\prod_{i=0}^{n-1} \gamma_i = -1$ and $\{A_i\}_{i=0}^{n-1}$ are the measurements labelling each vertex. We start with a few instantiations of particular cycle cases.

\subsection{The $4$-cycle scenario}
\label{n=4}

We start off by constructing operators corresponding to the inequalities of the 4-cycle:
$$ \sum_{i=0}^{3}\gamma_{i}\langle A_{i}A_{i\oplus 1}\rangle \overset{\textit{NCHV}}{\leq} 2 $$
The relevant quantum operator form of the LHS in terms of $m$-qubit Pauli operators is as follows:
$$\langle \Gamma \rangle = \gamma_{0}\langle P_{0}P_{1}\rangle + \gamma_{1}\langle P_{1}P_{2}\rangle + \gamma_{2}\langle P_{2}P_{3}\rangle + \gamma_{3}\langle P_{3}P_{4}\rangle $$
Remember that we introduced the edge Paulis previously, and so the operator $\Gamma$ can be written in terms of them as follows:
$$\Gamma = \gamma_{0}L_{0}+\gamma_{1}L_{1}+\gamma_{2}L_{2}+\gamma_{3}L_{3}$$
The condition for contextuality of the Paulis then becomes:
$$\langle \Gamma \rangle > 2 $$
So, if we can find some that $\ket{\Psi}$ that gives the expectation value above 2 means those statistics don't have a non-contextual hidden variable explanation. We square the operator obtained above:
$$\Gamma^{2}=(\gamma_{0}L_{0}+\gamma_{1}L_{1}+\gamma_{2}L_{2}+\gamma_{3}L_{3})^2 $$
For brevity, we write $\Gamma$ as:
$$\Gamma = a + b$$
where $a=\gamma_{0}L_{0}+\gamma_{1}L_{1}, b= \gamma_{2}L_{2}+\gamma_{3}L_{3}$.\\\\
Therefore $\Gamma^2={(a + b)^2}={a}^2+{b}^2+ ab + ba $. Here,
$$a^2=\gamma_{0}^{2}L_{0}^{2} + \gamma_{1}^{2}L_{1}^2+ \gamma_{0}\gamma_{1}\{L_{0},L_{1}\}$$
$$b^2=\gamma_{2}^{2}L_{2}^{2} + \gamma_{3}^{2}L_{3}^2+ \gamma_{2}\gamma_{3}\{L_{2},L_{3}\}$$
$$ab=\gamma_{0}\gamma_{2}\{L_{0},L_{2}\} + \gamma_{0}\gamma_{3}\{L_{0},L_{3}\} $$ 
$$ba= \gamma_{1}\gamma_{2}\{L_{1},L_{2}\} + \gamma_{1}\gamma_{3}\{L_{1},L_{3}\} $$
Using ${\gamma_i}^2=1$ and further the edge Pauli commutation(anti) relations derived in Section~\ref{commutativity}: $\{L_0,L_1\}=\{L_0,L_3\}=0=\{L_2,L_3\}=\{L_2,L_1\}$ and $[L_0,L_2]=0=[L_1,L_3]$, implies that:
$$a^2 = 2I $$
$$b^2= 2I$$
$$ab+ba= 2\gamma_{0}\gamma_{2}(L_{0}L_{2}) + 2\gamma_{1}\gamma_{3}(L_{1}L_{3})$$
$$\Gamma^{2}= 4I + 2\gamma_{0}\gamma_{2}L_{0}L_{2} + 2\gamma_{1}\gamma_{3}L_{1}L_{3} $$
We can further simplify this expression by noting that the commutation relations of the 4-cycle Paulis follow the following property:
$$L_{1}L_{3}=P_1P_2P_3P_0=-P_0P_1P_2P_3=-L_{0}L_{2}$$
This simplifies $\Gamma^2$ as follows:
$$\Gamma^2= 4I + 2L_{1}L_{3}({\gamma_0\gamma_2 - \gamma_1\gamma_3})$$
We know that there are only odd number of indices $i$ s.t. $\gamma_i = -1$, hence, for the 4-cycle, two cases exist: (i) only one index $i$ s.t. $\gamma_i = -1$ (ii) three indices $i$ s.t.$\gamma_i = -1$.
\\
If (i) is true, then:
$$\Gamma^2=4I+4L_{1}L_{3}= 4(I+L_{1}L_{3}) \hspace{1.0cm}  \text{(where $\gamma_{1}$ or $\gamma_{3}$ is -1)}$$
$$\Gamma^2=4I-4L_{1}L_{3}= 4(I-L_{1}L_{3}) \hspace{1.0cm}  \text{(where $\gamma_{0}$ or $\gamma_{2}$ is -1)}$$
If (ii) is true, then:
$$\Gamma^2=4I+4L_{1}L_{3}= 4(I+L_{1}L_{3}) \hspace{1.0cm}  \text{(where $\gamma_{1}$ or $\gamma_{3}$ is 1)}$$
$$\Gamma^2=4I-4L_{1}L_{3}= 4(I-L_{1}L_{3}) \hspace{1.0cm}  \text{(where $\gamma_{0}$ or $\gamma_{2}$ is -1)}$$ We get equivalent conditions from both possibilities. Also, note that $L_{1}L_{3}$ produces some Pauli operator in $\mathcal{P}_m$. Therefore, for every case the maximum eigenvalue of $\Gamma^2$ is 8. The maximum eigenvalue of $\Gamma$ is $\sqrt{8}$ i.e. $+ 2\sqrt{2}$ or $-2\sqrt{2}$. When $\Gamma = -2\sqrt{2}$, for a given inequality, then $\Gamma = 2\sqrt{2}$ for a valid inequality which is a negative multiple of the other one. Moreover, the state corresponding to this maximal violation is an eigenstate of the Pauli $L_{1}L_{3}=-P_0P_1P_2P_3$.
Therefore, each Pauli realization of a 4-cycle maximally violates some 4-cycle NC inequality.

\subsection{The $5$-cycle case}
\label{n=5}
In the 5-cycle case the quantum operator ($\Gamma$) appears as follows:
$$\Gamma = \gamma_{0}L_{0}+\gamma_{1}L_{1}+\gamma_{2}L_{2}+\gamma_{3}L_{3} + \gamma_{4}L_{4}$$ Squaring the operator, gives:
$$\Gamma^2=5I + \sum_{\substack{i=0, k=0 \\ i \neq k}}^{4} \gamma_i\gamma_k\{L_i,L_k\}$$ We know from the commutation constraints highlighted in Section~\ref{commutativity} that for $k=5$, all pair of distinct edge Paulis anti-commute, hence:
$$\Gamma^2 = 5I$$ This means that the maximum eigenvalue of $\Gamma$ is $\sqrt{5} < 3$. Therefore, no $m$-qubit Pauli realization for a 5-cycle produces a contextual behaviour, for any $m$.

\subsection{General $n \geq 4$-cycle scenarios}
\label{n>4}
Turns out that we can capture all cases above the 4-cycle scenario in a general manner. We capture this general situation in the form of a theorem.
\begin{theorem}
\label{no_contextuality}
    No $m$-qubit Pauli realization can witness contextuality in an $n > 4$-cycle scenario.
\end{theorem}
\begin{proof}
From the specific examples that we saw above we can notice that once we square the operator $\Gamma$, the only terms that survive among the anti-commutators are the ones where the edge Paulis commute with each other. For any $n \geq 5$ cycle:
$$\Gamma = \sum_{i=0}^{n-1} {\gamma_i}{L_i}$$
The condition for contextuality is $\langle \Gamma \rangle_{\ket{\Psi}} > n-2$.
$$\Gamma^2= kI + \sum_{\substack{i=0, k=0 \\ i \neq k}}^{n-1} \gamma_i\gamma_k\{L_i,L_k\}  $$ As noted in the 5-cycle case, within the summation on the r.h.s the anti-commuting pairs don't contribute whereas each commuting term appears as $2\gamma_i\gamma_k{L_i}{L_k}$. This means that:
$$\Gamma^2= kI+ \sum_{\substack{i=0, k=0 \\ i \neq k}}^{n-1} 2\gamma_i\gamma_k L_iL_k$$
Now we only need the number of such surviving terms to finish our proof. Using Fig.~\ref{generic_cycle} as an illustration, we count as follows: 
\begin{enumerate}
    \item $L_0,L_1,L_2$ each commutes with $n-5$ other edge paulis .
    \item For $L_3$ we need to avoid redundancy and discount any commutation with Paulis in 1. We have  $(n-5)-1 = n-6$ commutations i.e. we excluded one with $L_0$.
    \item For $L_4$ we need to exclude the commutation with $L_0$ \& $L_1$. Hence $(n-5)-2 = n-7$ such relations.
    \item We keep going like this, we reach $L_{n-4}$ where only 1 commutation relation needs to be counted i.e. with ${L_{n-1}}$.
    \item Beyond that every edge Pauli commutation combination has already been counted for in the steps above.
\end{enumerate}
This means that the total unique counts that contribute in the r.h.s (summation part) of $\Gamma^2$ above are:
\begin{equation}
\begin{split}
    3(n-5)+ (n-6) + (n-7) + (n-6) + \\ \ldots + 1 = 2(n-5) + \frac{(n-5)(n-4)}{2}
    \end{split}
\end{equation}
Clearly, this sum only makes sense for $n \geq 5$.
Now, we try to derive an upper bound on the maximum eigenvalue of $\Gamma^2$ operator defined above:
$$\langle \Gamma^2 \rangle _{\ket{\Psi}} = n + \sum_{i \neq k} 2\gamma_i\gamma_k  \langle L_iL_k \rangle_{\ket{\Psi}} $$ If we try to compute the algebraic upperbound of r.h.s above by noting that each term in the summation is $\leq 1$ (can't all together be 1 since all the Paulis across terms don't pairwise commute), Therefore:
$$\langle \Gamma^2 \rangle _{\ket{\Psi}} < n + 2 \left( 2(n-5) + \frac{(n-5)(n-4)}{2} \right) = n^2 - 4n$$ Therefore,
$$0 \leq \text{eig.val.}(\Gamma^2)  < n^2 - 4n $$
where eig.val.($\Gamma^2$) $\equiv$ any eigenvalue of operator $\Gamma^2$. This also means that:
$$-\sqrt {n^2 - 4n} < \text{eig.val.}(\Gamma) < \sqrt {n^2 - 4n}$$
For non-contextuality it must be that for all states $\ket{\Psi}$ in $\mathcal{H}_2^{\otimes m}$:
$$\langle \Gamma \rangle_{\ket{\Psi}} \leq (n-2)$$
Hence, condition for non-contextuality means that:
$$\sqrt{n^2 - 4n} < n-2$$
which always holds true.
Hence $\langle \Gamma \rangle_{\ket{\Psi}} < n-2 $ ($\forall$ $n \geq 5$ and $\ket{\Psi})$. This means that the statistics obtained from any faithful $m$-qubit Pauli realization of an $n$-cycle, lies strictly inside the classical (NC) polytope ($\forall$ $ n \geq 5$).
\end{proof}
\section{General compatibility graphs}
\label{Arbitrary_cases}
Having fully characterized Pauli contextuality in $n$-cycle scenarios, we now move our discussion towards more general scenarios. Owing to \Vorobev’s theorem, which states that an arbitrary KS scenario can witness contexuality iff it has an induced $n$-cycle scenario ($n \geq 4$), and our result that only the 4-cycle scenario produces contextual correlations by qubit Pauli operators, we now answer whether the presence of an induced 4-cycle is a necessary and sufficient condition for Pauli contextuality in arbitrary KS scenarios.
\subsection{Implication of 4-cycle contextuality}
The proofs above imply that if an arbitrary compatibility graph has at least one 4-cycle as its induced subgraph, then there always exists a quantum state that produces statistics indescribable by any noncontextual model. We can prove this by the following argument:
Imagine that a graph $G$ has an induced 4-cycle and assume that the statistics corresponding to Pauli realizations of $G$ are all non-contextual. This means that a joint probability distribution (JPD) over the outcomes of measurements realizing $G$ explains all the statistics one can observe over any sub scenario of $G$, by marginalizing over the rest of the graph. But from our result, on 4-cycle Pauli contextuality above, we know that we can find some quantum state $\ket{\Psi}$ for every Pauli realization such that it violates one of the 4-cycle NC inequalities, or its lifted version. 
This means that for each 4-cycle Pauli realization, there always exists statistics that cannot be explained by a JPD over $G$. Hence, the initial assumption is contradicted and graph $G$ can produce contextuality \footnote{More than anything, this proof underlines that adding more Paulis in addition to the 4-cycle in no way nullifies (or decreases) the contextuality witnessing effects of the 4-cycle.} via qubit Pauli operators.

This raises a natural question: Given an arbitrary Pauli compatibility graph, with at least one induced $n$-cycle ($n \geq 5$), can it witness contextuality? For graphs where there is exactly one such cycle, it's not complicated to construct a JPD, but for generic graphs one has to be careful. The same question further raises another important question: What is the precise role of induced cycles within compatibility graphs in producing contextuality? We know that presence of cycles ensures that one can always find some quantum observables such that the whole graph produces at least as much contextuality as any induced cycle \cite{xu2019necessary} whereas for the current case of the Pauli compatibility graph we know that the induced Pauli cycles independently can't produce contextuality. But this still doesn't imply that such a Pauli graph cannot exhibit contextuality. Therefore, we now move on to study Pauli compatibility graphs that are composed combinations of two 5-cycles.
\subsection{Two induced 5-cycles}
We consider three cases where two 5-cycles conjoin together differently.
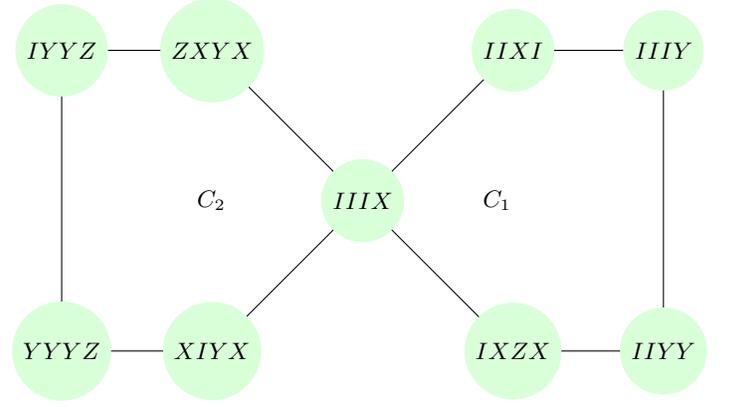
\begin{figure}
\begin{center}
\begin{tikzpicture}
[roundnode1/.style={circle, fill=green!15, minimum size=1mm}]
\node[text width=0.8cm] at (2,0) 
    {$C_1$};
\node[text width=0.4cm] at (-2,0) 
    {$C_2$};
\draw (0,0)node[roundnode1]{$IIIX$}--(-2,2)node[roundnode1]{$ZXYX$};
\draw (-2,2)node[roundnode1]{$ZXYX$}--(-4,2)node[roundnode1]{$IYYZ$};
\draw (-4,2)node[roundnode1]{$IYYZ$}--(-4,-2)node[roundnode1]{$YYYZ$};
\draw (-4,-2)node[roundnode1]{$YYYZ$}--(-2,-2)node[roundnode1]{$XIYX$};
\draw (-2,-2)node[roundnode1]{$XIYX$}--(0,0)node[roundnode1]{$IIIX$};
\draw (0,0)node[roundnode1]{$IIIX$}--(2,2)node[roundnode1]{$IIXI$};
\draw (2,2)node[roundnode1]{$IIXI$}--(4,2)node[roundnode1]{$IIIY$};
\draw (4,2)node[roundnode1]{$IIIY$}--(4,-2)node[roundnode1]{$IIYY$};
\draw (4,-2)node[roundnode1]{$IIYY$}--(2,-2)node[roundnode1]{$IXZX$};
\draw (2,-2)node[roundnode1]{$IXZX$}--(0,0)node[roundnode1]{$IIIX$};
\end{tikzpicture}
\caption{Two 5-cycles joined together at a node : a graph achievable only with $m \geq 4$. No such arrangement of Paulis can ever produce contextuality.}
\label{one_node}
\end{center}
\end{figure}

\textbf{Case I.} Only one node is common as depicted in Figure~\ref{one_node}. Clearly, such a Pauli-graph always has a non-contextual model. This is because a JPD ($p_T$) over it exists:
$$p_T= \frac{p_{C_1}p_{C_2}}{p_{IIIX}}$$ where $p_{C_i}$ is a JPD for $i^{th}$ cycle and ${p_{IIIX}}$ is probability for outcomes of measurement $IIIX$. Due to no-disturbance condition the common measurement $IIIX$ will always admit a unique probability distribution.

\begin{figure}
\begin{center}
\begin{tikzpicture}
[roundnode1/.style={circle, fill=green!15, minimum size=5mm}]
\node[text width=0.8cm] at (1.5,0) 
    {$C_1$};
\node[text width=0.4cm] at (-1.3,0) 
    {$C_2$};
\draw (0,2)node[roundnode1]{$IIX$}--(-2,2)node[roundnode1]{$XZX$};
\draw (-2,2)node[roundnode1]{$XZX$}--(-4,0)node[roundnode1]{$IYY$};
\draw (-4,0)node[roundnode1]{$IYY$}--(-2,-2)node[roundnode1]{$IIY$};
\draw (-2,-2)node[roundnode1]{$IIY$}--(0,-2)node[roundnode1]{$IXI$};
\draw (0,-2)node[roundnode1]{$IXI$}--(0,2)node[roundnode1]{$IIX$};
\draw (0,2)node[roundnode1]{$IIX$}--(2,2)node[roundnode1]{$XYX$};
\draw (2,2)node[roundnode1]{$XYX$}--(4,0)node[roundnode1]{$YYZ$};
\draw (4,0)node[roundnode1]{$YYZ$}--(2,-2)node[roundnode1]{$IIZ$};
\draw (2,-2)node[roundnode1]{$IIZ$}--(0,-2)node[roundnode1]{$IXI$};
\draw (0,-2)node[roundnode1]{$IXI$}--(0,2)node[roundnode1]{$IIX$};
\end{tikzpicture}
\caption{Two 5-cycles joined together : a graph achievable when $m \geq 3$. Such a graph can't produce contextuality for any realization of Pauli operators.}
\label{2_nodes}
\end{center}
\end{figure}
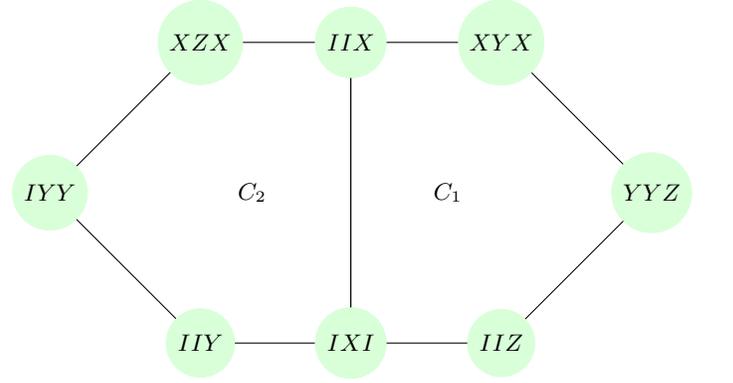
\textbf{Case II.} Two cycles share an edge as depicted in Figure~\ref{2_nodes}. This Pauli-graph always possesses a non-contextual model due to the existence of a JPD ($p_T$) over it:
$$p_T = \frac{p_{C_1}p_{C_2}}{p_{\{IIX,IXI\}}}$$ where $p_{C_i}$ is a JPD for $i^{th}$ cycle and $p_{\{IIX,IXI\}}$ is JPD over the common context $\{IIX,IXI\}$. Due to the no-disturbance condition, this common context will always admit a unique probability distribution.

\begin{figure*}
\begin{center}
\begin{tikzpicture}
[roundnode1/.style={circle, fill=green!15, minimum size=5mm}]
\node[text width=0.4cm] at (0,-3) 
    {(a)};
\node[text width=0.8cm] at (1.3,1.8) 
    {$1$};
\node[text width=0.8cm] at (2.5,0) 
    {$2$};
\node[text width=0.8cm] at (1.3,-1.8) 
    {$3$};
\node[text width=0.8cm] at (0.5,-1) 
    {$4$}; 
\node[text width=0.8cm] at (0.5,1) 
    {$5$};   
\node[text width=0.8cm] at (-1.2,1.7) 
    {$6$}; 
\node[text width=0.8cm] at (-2.1,0) 
    {$7$}; 
\node[text width=0.8cm] at (-1.3,-1.7) 
    {$8$};      
\begin{scope}
\draw (0,0)node[roundnode1]{$P_5$}--(0,2)node[roundnode1]{$P_1$};
\draw (0,2)node[roundnode1]{$P_1$}--(-2,1)node[roundnode1]{$P_7$};
\draw (-2,1)node[roundnode1]{$P_7$}--(-2,-1)node[roundnode1]{$P_6$};
\draw (-2,-1)node[roundnode1]{$P_6$}--(0,-2)node[roundnode1]{$P_4$};
\draw (0,-2)node[roundnode1]{$P_4$}--(0,0)node[roundnode1]{$P_5$};
\draw (0,2)node[roundnode1]{$P_1$}--(2,1)node[roundnode1]{$P_2$};
\draw (2,1)node[roundnode1]{$P_2$}--(2,-1)node[roundnode1]{$P_3$};
\draw (2,-1)node[roundnode1]{$P_3$}--(0,-2)node[roundnode1]{$P_4$};
\end{scope}    
\begin{scope}[xshift=7cm]
\node[text width=0.4cm] at (0,-3) 
    {(b)};
\draw (0,0)node[roundnode1]{$XI$}--(0,2)node[roundnode1]{$IX$};
\draw (0,2)node[roundnode1]{$IX$}--(-2,1)node[roundnode1]{$YX$};
\draw (-2,1)node[roundnode1]{$YX$}--(-2,-1)node[roundnode1]{$ZY$};
\draw (-2,-1)node[roundnode1]{$ZY$}--(0,-2)node[roundnode1]{$IY$};
\draw (0,-2)node[roundnode1]{$IY$}--(0,0)node[roundnode1]{$XI$};
\draw (0,2)node[roundnode1]{$IX$}--(2,1)node[roundnode1]{$ZX$};
\draw (2,1)node[roundnode1]{$ZX$}--(2,-1)node[roundnode1]{$YY$};
\draw (2,-1)node[roundnode1]{$YY$}--(0,-2)node[roundnode1]{$IY$};
\end{scope}
\end{tikzpicture}
\caption{Two 5-cycles sharing two edges $\{P_1,P_5\},\{P_4,P_5\}$: a graph achievable $\forall$ $m \geq 2$: (a) represents the generic $m$-qubit Paulis and respective context label on the edges. (b) exemplifies one realization with Paulis that violates a noncontextuality inequality of the scenario.}
\label{violation_producing_cases}
\end{center}
\end{figure*}
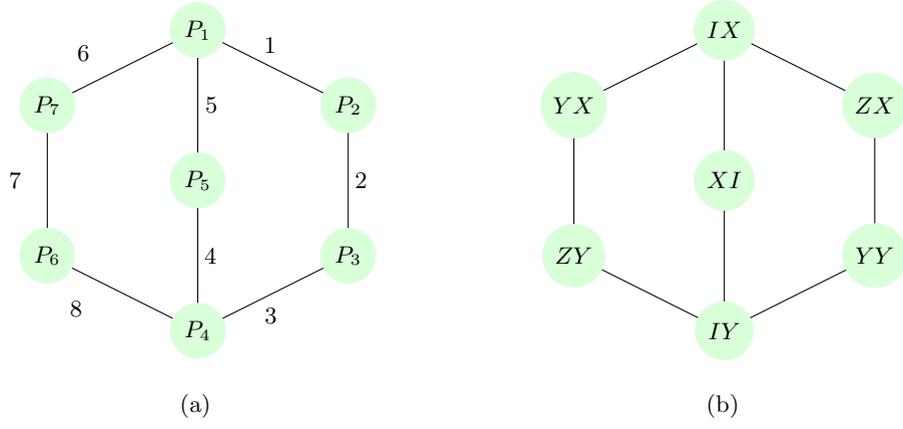
\textbf{Case III.} The two 5-cycles conjoin together with two edges in common. This case becomes non-trivial because now there is no guarantee that the JPDs over the two cycles give the same marginals for the unmeasurable correlation, see $\{P_1,P_4\}$ in Figure~\ref{violation_producing_cases}: Due to the no-disturbance condition the JPDs over the two cycles are constrained to give the same overlap for the common contexts $\{P_1,P_5\}$ and $\{P_4,P_5\}$ in Figure~\ref{violation_producing_cases}. But since $\{P_1,P_4\}$ is not a context, the same doesn't hold for it. 

\hspace{0.5cm}To test whether Pauli operators produce contextuality we first obtained all the noncontextuality inequalities corresponding to this graph, using the PORTA software package \cite{porta}. Then, using the qubit Pauli operators illustrated in Fig.~\ref{violation_producing_cases}(b), we observed a violation of the following noncontextuality inequality:
\begin{multline}
-p_{--}^{1}+ p_{-+}^{2} + p_{-+}^{3}- p_{--}^{4}+(p_{-+}^{5}+p_{+-}^{5}+2p_{--}^{5})-\\ 
p_{--}^{6}+(p_{-+}^{7}+p_{+-}^{7}+2p_{--}^{7})+
(p_{+-}^{8}-p_{-+}^{8}+ p_{--}^{8}) \leq 3 
\end{multline}
Here $p_{ab}^{i}$ represents probability of joint outcome $ab$ on measuring the $i^{th}$ context. By translating the projectors of contexts to Pauli operators, we can translate this inequality into an operator inequality, as follows:
\begin{equation}
\begin{split}
-(P_{1}P_{2}+P_{2}P_{3}+P_{3}P_{4}+P_{4}P_{5}+P_{1}P_{7}) + P_{4}P_{6} - \\ (P_{4}+ P_{5} + P_{6}+P_{7}) \leq 4I
\end{split}
\end{equation}
Using the Pauli realization in Fig.~\ref{violation_producing_cases}(b), the l.h.s above becomes: 
$$-(ZI+XZ+YI+XY+YI)+ ZI - (IY + XI + ZY + YX)$$
Now, if we check for the maximum eigenvalue of this expression, it turns out to be $4.2716 > 4$. Hence, a violation of the inequality. The corresponding state is:
$(0.2787 - 0.5952i,-0.2787 - 0.3342i,-0.4092 + 0.1482i,-0.4352 + 0.0000i)^{T}$ in the computational basis. This implies that the distribution for the context $\{P_1,P_4\}$ obtained from the graph, via marginalisation, won't be unique. Therefore, these conjoined cycles produce contextuality. This illustrates that the quantum violation of the NC inequalities for a given scenario (compatibility graph) doesn't necessarily accompany the violation of some induced cycle NC inequality within the graph. So, it seems that the fundamental role of an induced cycle in a compatibility graph of a fixed set of quantum measurements is only to preclude application of \Vorobev's theorem to the graph. The chosen set of quantum measurements then may or may not produce contextuality.

Another way to understand this is in terms of the geometrical approach to NC correlations i.e.~the correlation polytopes. In Ref.~\cite{choudhary2024lifting}, the authors showed that a facet-defining inequality for a KS scenario remains a facet-defining inequality for any scenario that is an extension of it. Therefore, for scenarios that can be seen as extensions of cycle scenarios, the $n$-cycle noncontextuality inequalities are just a proper subset of all facet-inducing inequalities of the scenario. The same holds for the case presented in case III above. Pauli contextuality observed is then just the violation of a noncontextuality inequality outside this proper subset.

\end{document}